\documentclass[11pt]{article}
\usepackage{fullpage}
\usepackage{epsfig}
\usepackage{graphics}
\usepackage{latexsym}
\usepackage{amsmath}
\usepackage{amsfonts}
\usepackage{amssymb}
\usepackage{mathrsfs}
\usepackage{amsthm}
\usepackage{xspace}
\usepackage{epstopdf}
\usepackage{float}
\usepackage{hyperref}
\usepackage{pgfplots}
\usepackage{thmtools}
\usepackage{thm-restate}

\numberwithin{equation}{section}
\usepackage[ruled,vlined]{algorithm2e}
\SetArgSty{textrm}

\usepackage[english]{babel}
\usepackage[nottoc]{tocbibind}

\usepackage{caption}
\usepackage{subcaption}

\usepackage[usenames,dvipsnames]{xcolor}
\usepackage{pifont}
\usepackage{booktabs}
\usepackage{multirow}

\newcommand{\bluecomment}[1]{\textcolor{blue}{\textrm{#1}}}

\newcommand{\ma}[1]{\textcolor{magenta}{\textrm{#1}}}
\newcommand{\redcomment}[1]{\textcolor{red}{\textrm{#1}}}

\def\Top{\mathrm{Top}}
\def\CM{\mathrm{CM}}
\def\SC{\mathrm{SC}}
\def\OPT{\mathrm{OPT}}

\usepackage{amsthm}     % for using \theoremstyle{plain}, {definition}, {remark}
\theoremstyle{plain}     % default. it sets the text in italic and adds extra space above and below the
\newtheorem{theorem}{Theorem}

\newtheorem{conjecture}{Conjecture}

\newtheorem{lemma}{Lemma}

\theoremstyle{definition} % adds extra space above and below, but sets the text in roman.
\theoremstyle{remark} % Font is set in roman, with no additional space above or below.
\makeatletter
\@addtoreset{equation}{section}
\def\section{\@startsection {section}{1}{\z@}{-3.5ex plus -1ex minus
 -.2ex}{2.3ex plus .2ex}{\large\bf}}
\makeatother

\def\bfm#1{\mbox{\boldmath$#1$}}

\def\0{\bfm 0}

\newcommand{\med}{\text{med}}

\DeclareMathAlphabet{\mathpzc}{OT1}{pzc}{m}{it}

\newcounter{my}

\newcounter{my2}

\newcounter{my3}

\newcounter{my4}

\newcounter{my5}

\newcounter{my6}

\allowdisplaybreaks 

\begin{document}

\title{Strategyproof Mechanisms for Euclidean Facility Location Problems under $L_p$-norm Social Cost}

\date{}
\maketitle

\vspace{-3em}
\begin{center}

\author{Hau Chan$^{1}$\quad Jianan Lin$^{2}$\quad Chenhao Wang $^{3,4}$\\
${}$\\
$1$ University of Nebraska-Lincoln\\
$2$ Rensselaer Polytechnic Institute\\
$3$ Beijing Normal University-Zhuhai\\
$4$ Beijing Normal-Hong Kong Baptist University\\
\medskip
}
\end{center}

\begin{abstract}
We study strategyproof mechanisms for eliciting agents' location preferences truthfully in the Euclidean plane $\mathbb R^2$ and locating a facility so as to minimize the $L_p$-norm social cost, defined as the $L_p$-norm of the vector of distances from the facility to the agents’ preferred locations,
 for any $p \ge 1$. 
While the cases  $p=1$ and $p=\infty$ have been well-studied, open questions remain about the optimal approximation ratios achievable by strategyproof mechanisms for general 
$p$.

Our first result resolves an open question of Goel and Hann-Caruthers [Soc. Choice Welf. 2023]. 
They showed that the coordinate-wise median (CM) mechanism achieves an approximation ratio lying  between \(2^{1-\frac{1}{p}}\) and \(2^{\frac{3}{2}-\frac{2}{p}}\) for $p\ge 2$, and they conjectured that it is exactly \(2^{1-\frac{1}{p}}\). 
We confirm this conjecture, and we further show that CM has a tight $\sqrt 2$-approximation for $1\le p\le 2$. Since it is previously known that the CM mechanism has the optimal approximation ratio among all deterministic anonymous strategyproof mechanisms for all $p\ge 1$, we complete the picture of deterministic mechanisms.

Our second and third results demonstrate that two randomized mechanisms can yield better approximation ratios.
In particular, we first consider the uniformly rotated coordinate-wise median (URCM) mechanism, 
and prove that, for \(1\le p<2\), its approximation ratio strictly improves over the deterministic bound \(\sqrt{2}\), while no such improvement is possible for $p\ge 2$. 
We then study the centroid random dictatorship  mechanism that returns the average location (i.e., centroid) and the random dictatorship each with half probability, 
and show that its approximation ratio strictly improves over CM and URCM for every finite \(p\gtrsim 1.6\). In fact, this approximation ratio holds for any dimensional Euclidean space. 
Moreover, our analysis independently recovers the classical deterministic and randomized results  for $p=1$ [Meir, SAGT 2019] [Barak, EC 2026] and $p=\infty$ [Goel and Hann-Caruthers, SCW 2023] [Tang et al., EC 2020] using significantly different techniques. 

We further show that the robustness of these mechanisms is not limited to \(L_p\)-norm social costs. Under every symmetric monotone norm objective, CM and URCM both obtain a \(2\)-approximation guarantee in the Euclidean plane. Moreover, the CRD analysis extends to arbitrary real normed vector spaces, where it yields a tight \(2-\frac1n\) approximation.
\end{abstract}

\section{Introduction}

\emph{Approximate mechanism design without money} aims to address optimization problems \cite{procaccia2013approximate} in which \emph{strategic} agents provide parts of the input to the optimization problems and objectives.  
Examples of these optimization problems include facility location \cite{procaccia2013approximate}, allocation of items \cite{guo2010strategy}, generalized assignment \cite{fadaei2017generalized}, voting \cite{filos2014truthful}, matching \cite{chen2011mechanism}, and scheduling \cite{koutsoupias2014scheduling}. 
Because strategic agents can often misreport input information to manipulate the solutions to optimization problems to their own benefit, the algorithmic or mechanism designer seeks to design \emph{strategyproof} mechanisms that elicit truthful input information from the agents and compute optimal solutions. 
Unfortunately, because the use of money is not allowed, strategyproof mechanisms for certain optimization problems must trade solution optimality for strategyproofness in order to incentivize agents to report information truthfully. 
That is, any strategyproof mechanism can only guarantee a certain approximation ratio (i.e., the ratio between the mechanism's solution objective value and the optimal value for any problem instance is bounded by that ratio). 
Therefore, a major research goal in approximate mechanism design without money is to design strategyproof mechanisms with approximation ratios as low as possible. 

%its output solution is only \emph{approximately optimal}). 
%attains an approximation ratio guarantee, and its solution can only be \emph{approximately optimal}. 
%the solutions from any strategyproof mechanisms can only be \emph{approximately optimal}, and the mechanisms provide solutions with approximation guarantees. 

\paragraph{The Canonical Case Study: Facility Location in $\mathbb R$.}
Procaccia and Tennenholtz \cite{procaccia2013approximate}, who first initiated the study of approximate mechanism design without money, examined simplified versions of facility location problems to illustrate the trade-off between solution optimality and strategyproofness in the design of strategyproof mechanisms. 
In the simplest facility location problem, a social planner seeks to determine an optimal location for a facility (e.g., a park, school, or library) on the real line $\mathbb R$ to serve a set of $n$ agents, where optimality is measured by a cost objective defined as either the (total) sum or the maximum of the distances between the facility location and the agents' preferred facility locations. 
Under the total cost, \cite{procaccia2013approximate} showed that the median mechanism, which locates the facility at the median of the agents' preferred facility locations, is both optimal and strategyproof. 
Unfortunately, under the maximum cost, \cite{procaccia2013approximate} showed that no strategyproof mechanism can achieve an approximation ratio better than 2. % (i.e., an impossibility lower bound result showing that the ratio cannot be lower than 2 with respect to the maximum cost). 
To complement the impossibility result, \cite{procaccia2013approximate} showed that the leftmost mechanism, which locates the facility at the leftmost of the agents' preferred locations, is strategyproof and achieves an approximation ratio of 2. 
To further improve the approximation ratio, \cite{procaccia2013approximate} proposed a randomized mechanism, which locates the facility at the leftmost, rightmost, and their average locations with probabilities $\frac14$, $\frac14$, and $\frac12$, respectively. The mechanism achieves an approximation of 1.5, which is the best any randomized strategyproof mechanism can achieve. 
Beyond the initial work, other studies examine the design of strategyproof mechanisms for settings that consider various multiple-facility locations \cite{fotakis2013strategyproof,fotakis2013winner}, location spaces (e.g., trees, cycles, or high-dimensional spaces) \cite{alon2010strategyproof,DBLP:conf/sagt/Meir19,GoelH23,gravin2025approximation}, alternative cost objectives \cite{feigenbaum2017approximately,GoelH23,walsh2025equitable}, and agent preferences \cite{kanellopoulos2026constrained,cheng2013strategy,fong2018facility,fang2025heterogeneous}. 
We refer readers to a survey \cite{chan2021mechanismsurvey} for more thorough coverage. 

\paragraph{Our Study: Euclidean Facility Location in $\mathbb R^2$.} 
One of the most natural and realistic extensions of the classical settings is the Euclidean facility location problem, in which facility locations and agents' preferences lie in the Euclidean plane $\mathbb R^2$.  
Under the $L_1$-norm social cost (i.e., total cost) and $L_{\infty}$-norm social cost (i.e., maximum cost), the coordinate-wise median (CM) mechanism, which places the facility at the median of each coordinate, is known to have approximation ratios of \(\sqrt{2}\) \cite{DBLP:conf/sagt/Meir19} and 2 \cite{GoelH23}, respectively.
%Moreover, Gravin and Jia \cite{gravin2025approximation} showed that \ma{pending} the CM mechanism has a larger approximation ratio of $\sqrt{6\sqrt{3}-8}$ under the $L_1$-norm social cost, but their general (larger) approximation ratio results extend to arbitrary dimensions.
Leveraging randomization to improve, Barak \cite{barak2026facility} and Tang et al. \cite{tang2020characterization} provided randomized mechanisms with approximation ratios of \(\frac{4}{\pi}\) and \(2-\frac{1}{n}\) under the  $L_1$-norm and $L_{\infty}$-norm social costs, respectively.

As for the general $L_p$-norm social cost with $p\ge 1$, Goel and Hann-Caruthers \cite{GoelH23} showed that the CM mechanism has the lowest approximation ratio among all deterministic, anonymous and strategyproof mechanisms, lying between \(2^{1-\frac{1}{p}}\) and \(2^{\frac{3}{2}-\frac{2}{p}}\) when \(p\ge 2\). They further conjecture that it is equal to \(2^{1-\frac{1}{p}}\).
Beyond these existing mechanism design results, no stronger results are known for the general $L_p$-norm social cost in two-dimensional Euclidean facility location problems. 
Therefore, in this paper, we strive to address the following main question:

\begin{quote}
%\emph{Do there exist strategyproof mechanisms that match or obtain better approximation ratios for Euclidean facility location problems under the $L_p$-norm social cost?}
\emph{What are the optimal approximation ratios of deterministic or randomized strategyproof 
mechanisms for Euclidean facility location under the $L_p$-norm social cost?}
\end{quote}

We address the above question affirmatively by showing that the CM mechanism has exactly an approximation ratio  \(2^{1-\frac{1}{p}}\) (thus closing the gap left in \cite{GoelH23}) and there exist two randomized mechanisms that achieve better approximation ratios than \(2^{1-\frac{1}{p}}\) across various values of $p$. 
%As a result, we effectively address the conjecture stated in \cite{GoelH23} regarding \(2^{1-\frac{1}{p}}\) approximation ratio and independently recover the randomized mechanism results for $p=1$ and $p=\infty$ using significantly different techniques \cite{barak2026facility}. 

\subsection{Our Contributions}\label{subsec:contribution}

We study the design of deterministic and randomized strategyproof mechanisms for 2-dimensional Euclidean facility location problems under the \(L_p\)-norm social cost with $p\ge 1$.  The distance between two points $x,y\in\mathbb R^2$ is the Euclidean distance \(d(x,y)=\|x-y\|_2\) and we use $\|x-y\|$ for convenience.
Moreover, given a location profile $\mathbf{x}=(x_1, \ldots, x_n)\in(\mathbb R^2)^n$  of $n$ agents and a facility location \(y\), the \(L_p\)-norm social cost is defined as 
\[
    \mathrm{SC}_p(\mathbf{x},y)
    =
    \left(
        \sum_i \|x_i-y\|^p
    \right)^{\frac{1}{p}},
\]
with the usual interpretation
\(\mathrm{SC}_{+\infty}(\mathbf{x},y)=\max_i\|x_i-y\|\). 
Thus, \(p=1\) corresponds to the total cost, and \(p=+\infty\) corresponds to the maximum cost. 

We discuss three strategyproof mechanisms and their approximation ratios. 
Denote by $\alpha_p(f)$ the worst-case approximation ratio of mechanism $f$ under the $L_p$-norm social cost. % that address the question mentioned above. 
%We present the following results using formal statements from the latter part of the paper. %but without formal proofs. % statements are given using restatable environments. \hau{what does restatable environment mean?} 
%In the technical sections of the paper (Section \ref{sec:cm} - \ref{sec:crd}), we restate the mechanisms and theorems using the same numbers and labels and formally prove the claimed results. 

\paragraph{Coordinate-wise Median Mechanism.} We first consider a classic and most well-studied deterministic mechanism in social choice and facility location problems.

%More formally, the coordinate-wise median mechanism is defined as follows. 

\begin{restatable}[Coordinate-wise Median]{mechanism}{cmmechanism}
\label{mec:cm}
Given a profile \(\mathbf{x}=(x_1,\ldots,x_n)\in(\mathbb R^2)^n\), write
\(x_i=(x_{i,1},x_{i,2})\).
For any real numbers \(a_1,\ldots,a_n\), let \(\med(a_1,\ldots,a_n)\) be the \(\lceil \frac{n}{2}\rceil\)-th smallest number.
The coordinate-wise median mechanism, denoted by
\(\mathrm{CM}\), outputs
\[
    \mathrm{CM}(\mathbf{x})
    =
    \left(
        \med\left(x_{1,1},\ldots,x_{n,1}\right),
        \med\left(x_{1,2},\ldots,x_{n,2}\right)
    \right).
\]
\end{restatable}

%Our first result demonstrates that a deterministic mechanism achieves the factor \(2^{1-\frac{1}{p}}\), effectively addressing the open question in \cite{GoelH23}. Specifically, 
Goel and Hann-Caruthers \cite{GoelH23} showed that, for \(p\ge 2\), the approximation ratio of the CM mechanism lies between \(2^{1-\frac{1}{p}}\) and \(2^{\frac{3}{2}-\frac{2}{p}}\), and conjectured that it is exactly \(2^{1-\frac{1}{p}}\). 
%Meir \cite{DBLP:conf/sagt/Meir19} proved that it is a \(\sqrt{2}\)-approximation for \(p=1\). 
We close this gap and show that CM achieves the approximation ratio of \(2^{1-\frac{1}{p}}\) for all \(p\ge 2\). We further prove a tight $\sqrt{2}$-approximation for all $1\le p\le 2$ which extends the corresponding result of Meir \cite{DBLP:conf/sagt/Meir19} when $p=1$.
Together with the optimality of CM shown in \cite{GoelH23}, our results indicate that the best possible approximation ratio of any deterministic anonymous strategyproof mechanism is \(2^{1-\frac{1}{p}}\) for all \(p\ge 2\) and $\sqrt 2$ for $1\le p\le 2$.
The following theorem formally states the approximation ratios. 

\begin{restatable}{theorem}{cmtheorem}
\label{thm:cm-tight}
The coordinate-wise median mechanism achieves a tight approximation ratio under the \(L_p\)-norm social cost in $\mathbb R^2$:
\[
    \alpha_p(\mathrm{CM})
    =
    \begin{cases}
        \sqrt{2}, & 1\le p\le 2,\\[3pt]
        2^{1-\frac{1}{p}}, & 2\le p\le +\infty.\\[3pt]
    \end{cases}
\]
\end{restatable}

Very recently, concurrent and independent work by Hastings~\cite{hastings2026strategic}
also studied the coordinate-wise median mechanism under \(L_p\)-norm social cost.
Hastings considers the more general setting of arbitrary-dimensional
\(\ell_q\) spaces and, in particular, obtains tight bounds for all \(p,q\ge1\)
when \(d=2\). Our work was developed independently and, beyond resolving the
Euclidean-plane case, focuses on randomized mechanisms that improve over CM
for several ranges of \(p\).

\iffalse
\ma{pending}The proof separates the two regimes \(1\le p\le 2\) and \(p\ge 2\). For the former, we relate the \(\ell_2\) and \(\ell_p\) norms in the plane, which together with a standard one-dimensional median inequality yields a \bluecomment{bound $\sqrt 2$.} %\(2^{\frac{p}{2}}\) bound on the \(p\)-th power, giving \(\sqrt2\). 
For the latter, we center the instance at the median and pair agents with opposite coordinate signs; using the geometric fact that for \(a\cdot b\le0\) the \(\ell_2\) norm is superadditive in squares, we combine the triangle and H\"older inequalities to obtain the tight factor \(2^{1-\frac1p}\).
\fi

\paragraph{Uniformly Rotated Coordinate-wise Median Mechanism.}
We formalize and analyze a random-rotation idea suggested by Goel and Hann-Caruthers \cite{GoelH23}, which is also independently explored recently by Barak \cite{barak2026facility} for $p=1$. 
The mechanism samples a uniformly random coordinate system and then applies the coordinate-wise median mechanism within the new coordinate system. 
More formally, for each angle \(\theta\in[0,2\pi)\), let $e_\theta=(\cos\theta,\sin\theta)$ and $e_\theta^\perp=(-\sin\theta,\cos\theta)$.
These two vectors form an orthonormal basis of \(\mathbb{R}^2\): \(e_\theta\) is the unit vector
pointing in direction \(\theta\), and \(e_\theta^\perp\) is the unit vector obtained by rotating
\(e_\theta\) counterclockwise by \(\frac{\pi}{2}\). Equivalently, \(\{e_\theta,e_\theta^\perp\}\) is the
standard coordinate system rotated by angle \(\theta\).
% The  mechanism %(also called RRCWM in \cite{barak2026facility}) 
%  first samples a rotation angle \(\Theta\sim\mathrm{Unif}[0,2\pi)\); then express every reported location \(x_i\in\mathbb R^2\) in the rotated coordinate system \(\{e_\Theta,e_\Theta^\perp\}\), take the median of the first coordinates and the median of the second coordinates, and finally map the resulting pair of medians back to a point in the plane.

\begin{restatable}[Uniformly Rotated Coordinate-Wise Median (\(\mathrm{URCM}\))]{mechanism}{urcmmechanism}
\label{mec:urcm}
Given a profile \(\mathbf{x}=(x_1,\ldots,x_n)\in(\mathbb R^2)^n\), sample
\(\Theta\sim\mathrm{Unif}[0,2\pi)\). For any \(x_i\), its coordinates in the rotated basis
\(\{e_\Theta,e_\Theta^\perp\}\) are
\(\bigl(\langle x_i,e_\Theta\rangle,\langle x_i,e_\Theta^\perp\rangle\bigr)\), where $\langle\cdot,\cdot\rangle$ is the inner product of vectors.
Let
\[
    m_{\Theta,1}
    =
    \med\left(
        \left\langle x_1,e_\Theta\right\rangle,\ldots,
        \left\langle x_n,e_\Theta\right\rangle
    \right),
    \qquad
    m_{\Theta,2}
    =
    \med\left(
        \langle x_1,e_\Theta^\perp\rangle,\ldots,
        \langle x_n,e_\Theta^\perp\rangle
    \right)
\]
be the median of the first and the second coordinates, respectively.
%where \(\med(\cdot)\) still denotes the median of a list of real numbers, with a fixed tie-breaking rule when \(n\) is even. 
The mechanism finally maps the resulting pair of medians back to the point in $\mathbb R^2$: %outputs the random point
\[
    \mathrm{URCM}(\mathbf{x})
    =
    m_{\Theta,1}e_\Theta + m_{\Theta,2}e_\Theta^\perp .
\]
\end{restatable}

%Our second result shows that there is a randomized mechanism that achieves better approximation ratios than \(2^{1-\frac{1}{p}}\) for \(1\le p<2\).

We show that this randomization strictly improves the approximation ratio $\sqrt 2$ of the CM mechanism for \(1\le p<2\), %including matching the approximation ratio of the standard social cost case \(p=1\) of \cite{barak2026facility}, 
while no improvement is obtained for \(p\ge 2\). The following theorem formally states the approximation ratios as Gamma-function expressions. % of the URCM mechanism for all values of $p\ge 1$. 

\begin{restatable}{theorem}{urcmtheorem}
\label{thm:urcm}
The uniformly rotated coordinate-wise median mechanism is randomized and strategyproof. Its approximation ratio
under the $L_p$-norm social cost in $\mathbb R^2$ is bounded by:
\[
    \alpha_p(\mathrm{URCM})
    \le
    \begin{cases}
        \displaystyle
        2\left(
            \dfrac{
                \Gamma\left(\frac{p+1}{2}\right)
            }{
                \sqrt{\pi}\,\Gamma\left(1+\frac{p}{2}\right)
            }
        \right)^{\frac{1}{p}},
        & 1\le p<2,\\[14pt]
        2^{1-\frac{1}{p}},
        & 2\le p\le +\infty.\\[5pt]
    \end{cases}
\]
In addition, for every \(1\le p\le 2\),
\[
    \alpha_p(\mathrm{URCM})
    \ge
    % 2^{1-\frac{1}{p}}
    \frac{2^{3-\frac{1}{p}}}{\pi}
    \int_0^{\frac{\pi}{4}}
    \left(\cos^p t+\sin^p t\right)^{\frac{1}{p}}\,dt.
\]
Consequently, the bounds are tight at \(p\in \{1, 2,+\infty\}\). In particular,
\[
    \alpha_1(\mathrm{URCM})=\frac{4}{\pi},
    \qquad
    \alpha_2(\mathrm{URCM})=\sqrt{2}
    \qquad
    \alpha_{+\infty}(\mathrm{URCM}) = 2.
\]
For \(2<p<+\infty\), the upper bound is asymptotically tight.
% The bounds for $p=1$ and \(2\le p<+\infty\) are asymptotically tight; and for $p=+\infty$, it is tight.
\end{restatable}

For \(p=1\), the bound is equal to \(2\frac{\Gamma(1)}{\sqrt{\pi}\Gamma(3/2)}=\frac{4}{\pi}\). As \(p\) approaches 2 from below, the bound increases  to \(\sqrt{2}\). Hence,
for every \(1\le p<2\), URCM strictly improves over the deterministic bound \(\sqrt{2}\). %the deterministic coordinate-wise median.
The key reason is that, for \(p < 2\), the \(p\)-th power of the two
coordinate projections of a vector depends on the choice of coordinate
axes. A bad fixed coordinate system may overestimate the contribution of
some directions, but URCM avoids committing to one such system by averaging
uniformly over all rotations.
% For a fixed vector \(v\), after a random rotation, the two projected lengths
% are proportional to \(|\cos \Theta|\) and \(|\sin \Theta|\). Therefore, the
% expected contribution of \(v\) is governed by the average value of
% \(|\cos \Theta|^p\) over a uniformly random angle \(\Theta\). This average is
% exactly
% \[
% \frac{\Gamma\left(\frac{p+1}{2}\right)}
% {\sqrt{\pi}\,\Gamma\left(1+\frac{p}{2}\right)},
% \]
% obtained by evaluating the integral of \(|\cos \theta|^p\) over the circle
% using the standard beta-gamma identity. This is where the Gamma-function
% expression in the approximation ratio comes from.
At \(p = 2\), this averaging advantage disappears: for every rotation, the
sum of the squared projected lengths is exactly the squared Euclidean length
of the vector. Thus randomizing the coordinate system no longer reduces the
relevant quantity. For \(p \ge 2\), every fixed rotated coordinate-wise median
already satisfies the deterministic bound \(2^{1-1/p}\), so averaging over
rotations cannot improve.
%This improvement comes from randomly rotating the coordinate axes before applying the coordinate-wise median. We expand the Euclidean distance in the rotated basis, apply the same one-dimensional median inequality per coordinate, and then take expectations over the uniform rotation, which reduces to computing \(\mathbb{E}[|\cos\Theta|^p]\) and yields the Gamma-function expression. For \(p\ge 2\),  for every fixed rotation the deterministic bound \(2^{1-\frac{1}{p}}\) holds, so the random rotation cannot improve it. %a lower-bound instance with a majority at the origin confirms asymptotic tightness.

Very recently, an independent work of Barak \cite{barak2026facility} also studied the URCM mechanism and %with a focus on the $L_1$-norm social cost in high-dimensional Euclidean spaces. 
 proved an upper bound \(\frac{4}{\pi}\) for the $L_1$-norm social cost in the plane $\mathbb R^2$. 
%We briefly clarify the relationship between our work and \cite{barak2026facility}.  The two papers independently analyze the same randomized mechanism in Euclidean spaces,  referred to as RRCWM in \cite{barak2026facility} and as URCM in this paper. Both papers  prove an upper bound of \(\frac{4}{\pi}\) for the \(L_1\)-norm social cost in the Euclidean plane. 
Aside from this common result, however, the main emphases of the two papers are quite different. 
Our work studies three mechanisms under the general \(L_p\)-norm social cost in $\mathbb R^2$, while 
\cite{barak2026facility} is primarily concerned with $L_1$-norm social cost in high-dimensional Euclidean spaces,  covering both the classical and learning-augmented settings.

\paragraph{Centroid Random Dictatorship Mechanism.}

Our third result shows that another randomized mechanism achieves better approximation ratios than the URCM mechanism for \(p\) larger than around $1.60$. 
In particular, we analyze the randomized mechanism of Feldman and Wilf \cite{feldman2013strategyproof} (on the line) and Tang et al. \cite{tang2020characterization} (in multi-dimensional normed vector spaces), which incorporates the notion of a centroid into random dictatorships as follows. 

 %for maximum cost, i.e., $p=+\infty$.

 %,  denoted by \(\mathrm{CRD}\)
%They did not give the mechanism a specific name; we call it centroid random dictatorship, denoted by \(\mathrm{CRD}\). 
%We derive approximation ratios for all \(p\ge 1\).

\begin{restatable}[Centroid Random Dictatorship (CRD)]{mechanism}{crdmechanism}
\label{mec:crd}
Given a profile \(\mathbf{x}=(x_1,\ldots,x_n)\in(\mathbb R^2)^n\), let
\(\bar x=\frac{1}{n}\sum_i x_i\) be the centroid. Return \(\bar x\) with probability
\(\frac{1}{2}\), and return each agent location \(x_j\) with probability
\(\frac{1}{2n}\).
\end{restatable}

Tang et al. \cite{tang2020characterization} showed that the CRD mechanism is strategyproof and achieves a tight \(\left(2-\frac{1}{n}\right)\)-approximation when $p=+\infty$ (i.e., maximum cost). 
 %refer to such a mechanism as the centroid random dictatorship \(\mathrm{(CRD)}\) mechanism and 
In the following we generalize it for all values of \(p\ge 1\). %including recovering the result of \cite{tang2020characterization}.  

\begin{restatable}{theorem}{crdtheorem}
\label{thm:crd}
 The centroid random dictatorship mechanism is randomized and
strategyproof. Its approximation ratio under the $L_p$-norm social cost in $\mathbb R^d$ ($d\ge 1$) is bounded by
%satisfies
\[
    \alpha_p(\mathrm{CRD})
    \le
    \begin{cases}
        \displaystyle
        \frac{1}{2}q_n^{\frac{2}{p}-1}
        \left(
            2^{\frac{2}{p}-1}+2^{\frac{1}{p}}
        \right),
        & 1\le p\le 2,\\[12pt]
        \displaystyle
        \frac{1}{2}
        \left(
            2^{1-\frac{2}{p}}q_n^{1-\frac{2}{p}}
            +
            2^{1-\frac{1}{p}}
        \right),
        & 2\le p\le +\infty,\\[12pt]
        % \displaystyle
        % 2-\frac{1}{n},
        % & p=+\infty.
    \end{cases}
\]
where  \(q_n=1-\frac{1}{n}\).
The endpoint bounds at \(p=1,+\infty\) are tight; the bound at $p=2$ is tight for every even $n$; and when \(n=2\) the bound is tight for every \(1\le p\le+\infty\). 
% \hau{we only have tight sentence here; what about for the other theorems? do you need to mention about setting qn?}
\end{restatable}

These bounds are derived through a norm-operator analysis. After translating the instance so that the optimal facility is at the origin,  we represent the profile by the agents’ translated locations. Then the expected cost of CRD naturally decomposes into two terms: the cost of the centroid part and the cost of the random-dictatorship part. We introduce two norm operators $T$ and $S$: the centering operator $T$ controls the first term, and the pairwise-difference operator $S$ controls the second term.

% The analysis rewrites the expected cost as \(\frac12\|TX\|_p + \frac12\|SX\|_p\), where \(T\) and \(S\) are the centering and pairwise-difference operators acting on the translated agent vectors. We obtain exact norm bounds for \(p=1,2,\infty\) and then interpolate between them via the Riesz–Thorin theorem, leading to the two-part formulas for \(1\le p\le2\) and \(p\ge2\). %The tightness examples use extreme one-dimensional configurations and, for \(n=2\), match the bound for all \(p\). 
% \redcomment{[move it later]}

\paragraph{Symmetric Monotone Norm Objectives.}
Although our main results focus on the \(L_p\)-norm social cost, we also
consider a broader class of objectives in Section~\ref{sec:general-norm}.
For a symmetric monotone norm \(g:\mathbb R^n\to\mathbb R_+\), we define
\(\SC_g(\mathbf{x},y)=g(D(\mathbf{x},y))\), where
\(D(\mathbf{x},y)\) is the vector of distances from the facility location
\(y\) to the agents. We show that, in the Euclidean plane, both CM and URCM
are \(2\)-approximations for every such objective. We further show that the
CRD guarantee is even more robust: for every symmetric monotone norm
objective, CRD achieves a tight \(\left(2-\frac1n\right)\)-approximation in
arbitrary real normed vector spaces. Here, a real normed vector space means a
real vector space \(V\) equipped with a norm \(\|\cdot\|\), where the distance
between two points is induced by \(d(x,y)=\|x-y\|\). Since every \(L_p\)-norm
on the distance vector is symmetric and monotone, these results provide
uniform fallback guarantees beyond the \(L_p\)-specific analyses, while the
sharper bounds for CRD in Theorem~\ref{thm:crd} rely on the Euclidean
structure.

\paragraph{Comparison.}
Table~\ref{tab:upper-bounds} summarizes the upper bounds of the three mechanisms, where \(q_n=1-\frac{1}{n}\).
Figure~\ref{fig:upper-bounds-large-n} visualizes these upper bounds when \(n\to+\infty\), where \(q_n\) converges to \(1\).

\begin{table}[!htb]
\centering
\small
\renewcommand{\arraystretch}{1.35}
\caption{Upper bounds on the approximation ratios \(\alpha_p\). }
\label{tab:upper-bounds}
\begin{tabular}{c|c|c|c|c}
\toprule
Mechanism
&
\(1\le p<2\)
&
\(p=2\)
&
\(2<p<+\infty\)
&
\(p=+\infty\)
\\
\midrule
\(\mathrm{CM}\)
&
\(\sqrt{2}\)
&
\(\sqrt{2}\)
&
\(\displaystyle 2^{1-\frac{1}{p}}\)
&
\(2\)
\\
\addlinespace
\(\mathrm{URCM}\)
&
\(\displaystyle
2\left(
    \frac{
        \Gamma\left(\frac{p+1}{2}\right)
    }{
        \sqrt{\pi}\,\Gamma\left(1+\frac{p}{2}\right)
    }
\right)^{\frac{1}{p}}\)
&
\(\sqrt{2}\)
&
\(\displaystyle 2^{1-\frac{1}{p}}\)
&
\(2\)
\\
\addlinespace
\(\mathrm{CRD}\)
&
\(\displaystyle
\frac{1}{2}q_n^{\frac{2}{p}-1}
\left(
    2^{\frac{2}{p}-1}+2^{\frac{1}{p}}
\right)\)
&
\(\displaystyle \frac{1+\sqrt{2}}{2}\)
&
\(\displaystyle
\frac{1}{2}
\left(
    2^{1-\frac{2}{p}}q_n^{1-\frac{2}{p}}
    +
    2^{1-\frac{1}{p}}
\right)\)
&
\(\displaystyle 2-\frac{1}{n}\)
\\
\bottomrule
\end{tabular}
\end{table}

\begin{figure}[!htb]
    \centering
    \includegraphics[width=0.9\textwidth]{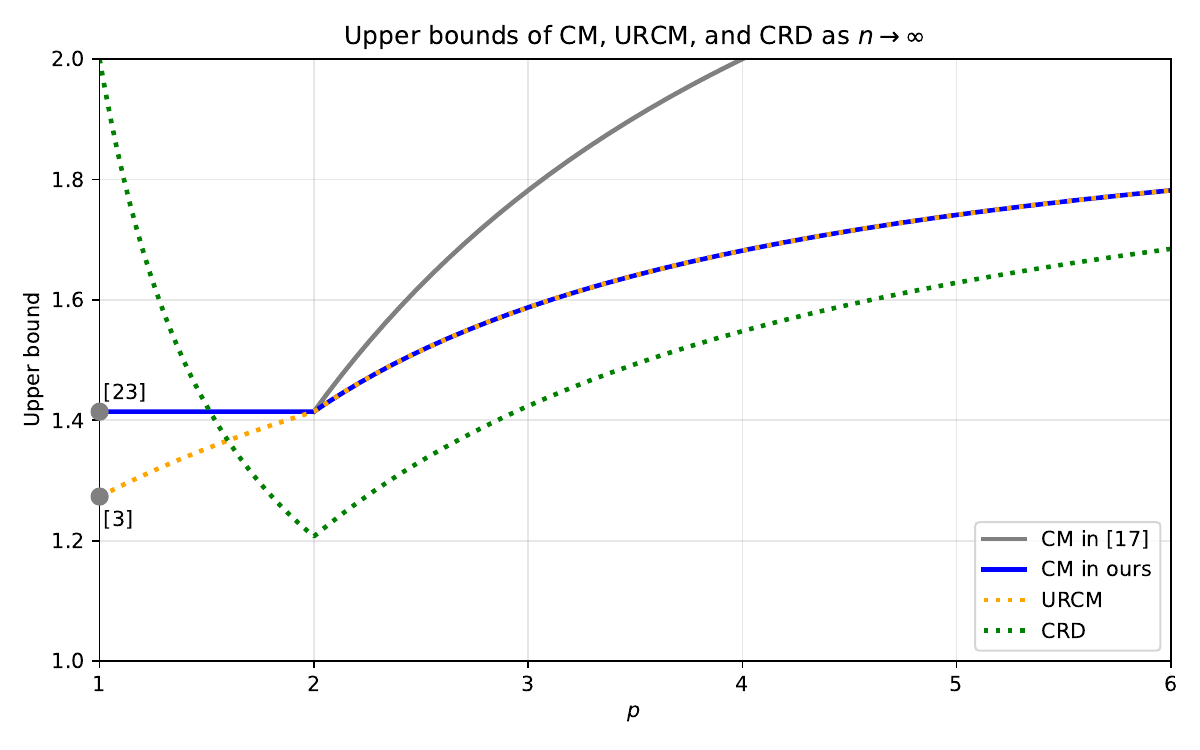}
    \caption{Upper bounds of \(\mathrm{CM}\), \(\mathrm{URCM}\), and
    \(\mathrm{CRD}\) with \(n\to+\infty\). %The horizontal axis is \(p\), and the vertical axis is the correspondin approximation upper bound. 
    The gray curve for $p\ge 2$ is by Goel and Hann-Caruthers \cite{GoelH23}, and the two gray points of CM and URCM when $p=1$ are by Meir \cite{DBLP:conf/sagt/Meir19} and Barak \cite{barak2026facility}, respectively. Our three curves converge to 2 as $p\to+\infty$.
    % The lower bound of deterministic anonymous strategyproof mechanisms is the same as the upper bound of CM.
    }
    \label{fig:upper-bounds-large-n}
\end{figure}

These bounds exhibit three distinct regimes. For small values of \(p\), namely \(1\le p<2\), URCM improves over the \(\sqrt 2\) guarantee of the coordinate-wise median by exploiting random rotations. As \(p\) approaches \(2\), CRD becomes more effective and eventually outperforms URCM, with the crossover occurring around \(p\approx 1.6\). For \(p>2\), CM and URCM have the same guarantee, while CRD gives a strictly better upper bound and approaches the known randomized guarantee for the maximum-cost objective. Thus, URCM is most useful near \(p=1\), whereas CRD is preferable near and above \(p=2\).

\medskip
\emph{Organizations.} Section~\ref{sec:model} presents preliminaries. Section~\ref{sec:cm} studies the CM mechanism. Section~\ref{sec:urcm} studies the URCM  mechanism. Section~\ref{sec:crd} studies the CRD  mechanism. We discuss the symmetric monotone norm in Section~\ref{sec:general-norm}.

\subsection{Additional Related Work}

% To the best of our knowledge, \cite{GoelH23} is the only work that considers Euclidean facility location problems under the $L_p$-norm social cost for general $p$, focusing solely on analyzing the approximation ratios of the coordinate-wise median mechanism. They also showed that this mechanism has the lowest worst-case approximation ratio among all deterministic, anonymous, and strategyproof mechanisms. 
% In addition to the results mentioned earlier, 
We briefly review multi-dimensional facility location problems and facility location problems under $L_p$-norm social cost within the context of approximate mechanism design without money.

\paragraph{Multi-dimensional Facility Location Problems.}
For Euclidean facility location problems in $\mathbb R^d$, Meir \cite{DBLP:conf/sagt/Meir19} showed that the coordinate-wise median (CM) mechanism achieves a \(\sqrt{d}\)-approximation under the total cost. 
Gravin and Jia \cite{gravin2025approximation} further improve the upper bound of the CM mechanism to a constant \(\sqrt{6\sqrt{3}-8}\approx 1.547\) in $\mathbb R^d$ for any dimension $d$. They also studied $d$-dimensional spaces with \(\ell_q\) distances (rather than Euclidean), establishing  nearly tight bounds for all $q\ge 1$, which are at most \(3\). 
By adding a uniformly random rotation before applying the CM mechanism, Barak \cite{barak2026facility}
showed that  the expected approximation ratio under the total cost in $\mathbb R^d$ can be improved to $[1.41-O(1/\sqrt{d}),1.547]$.
%This bound was later improved by Barak \cite{barak2026facility}, who proposed a randomized mechanism called RRCWM (which we refer to as URCM in this paper) and proved an approximation ratio of \(\frac{4}{\pi}\). 
In learning-augmented settings, Agrawal et al.~\cite{agrawal2022learning} studied a generalized CM mechanism with predictions for both the total cost and the maximum cost. Chan et al.~\cite{chan2025prediction} considered learning-augmented mechanisms for the maximum cost in two-dimensional spaces with \(\ell_q\) distances.
On the characterization side, Tang et al.~\cite{tang2020characterization} characterized group-strategyproof mechanisms in strictly convex spaces, while \cite{lin2020nearly} characterized strategyproof mechanisms for two agents in multi-dimensional spaces with \(\ell_q\) distances.

\paragraph{Facility Location Problems under \(L_p\)-norm Social Cost.}
Feigenbaum et al.~\cite{feigenbaum2017approximately} were the first to study facility location under the \(L_p\)-norm social cost systematically. They proved a tight approximation ratio \(2^{1-\frac{1}{p}}\) of deterministic mechanisms on the line for all $p\ge 1$. %and further showed that for two agents, the Left-Right-Middle (LRM) mechanism is optimal within a broad class of randomized strategyproof mechanisms. For special settings, 
Feldman and Wilf \cite{feldman2013strategyproof} studied the objective of minimizing the sum of squared costs on both lines and trees. In particular, they showed that the median mechanism achieves a \(2\)-approximation and that the CRD mechanism achieves a \(1.5\)-approximation on the line. Note that the sum of squared costs is the square of the \(L_2\)-norm social cost. Therefore, the square roots of their approximation ratios correspond to the ratios under the \(L_2\)-norm social cost.
%In their formulation, the objective is the sum of squared distances, under which the corresponding ratio is \(\frac{\sqrt{5}}{2}\); here we translate their result into our \(L_2\)-norm social cost formulation. 
%Fotakis and Tzamos \cite{fotakis2013strategyproof} considered concave cost functions, which can be viewed as complementary to the study of \(L_p\)-norm social costs.
In the variant of obnoxious facility location, Ye et al.~\cite{ye2015strategy} studied the \(L_2\) social utility, and Chan et al.~\cite{chan2025obnoxious} investigated both \(L_p\) social utility and \(L_p\) social cost.

\section{Preliminaries}\label{sec:model}

We study the facility location problem in the two-dimensional Euclidean
plane. Let \(N=\{1,2,\ldots,n\}\) be the set of agents. A location profile is
denoted by \(\mathbf{x}=(x_1,x_2,\ldots,x_n)\in\left(\mathbb R^2\right)^n\), where
\(x_i\in\mathbb R^2\) is the preferred location of agent \(i\). The distance
between two points is the Euclidean distance \(d(x,y)=\|x-y\|_2\). When there
is no ambiguity, we write \(\|\cdot\|\) for the Euclidean norm. For an agent
\(i\) and an alternative report \(x_i'\in\mathbb R^2\), we write
\((x_i',\mathbf{x}_{-i})\) for the profile obtained from \(\mathbf{x}\) by
replacing only agent \(i\)'s report with \(x_i'\). 

A deterministic mechanism is a function
\(f:(\mathbb R^2)^n\to\mathbb R^2\) that maps each reported profile to a facility location. A randomized mechanism is a function
\(f:(\mathbb R^2)^n\to\Delta(\mathbb R^2)\), where \(\Delta(\mathbb R^2)\) denotes the set of probability distributions over \(\mathbb R^2\). For a randomized mechanism, we write \(Y\sim f(\mathbf{x})\) for the random facility location sampled from the distribution returned on profile \(\mathbf{x}\). A deterministic mechanism can be viewed as a degenerate randomized mechanism concentrated on a single point.

Given a facility location $y\in\mathbb R^2$, the cost of agent $i$ is the distance between the facility location and the preferred location $x_i$, namely, $\|y-x_i\|$. For a random facility location
\(Y\), the expected cost of agent \(i\) is
\(\mathbb E_{Y}[\|x_i-Y\|]\).
A mechanism \(f\) is \emph{strategyproof} if no agent can reduce its expected cost by misreporting while the reports of all other agents are fixed. Formally, for every profile \(\mathbf{x}\in(\mathbb R^2)^n\), every agent \(i\in N\), and every alternative report \(x_i'\in\mathbb R^2\),
\[
\mathbb E_{Y\sim f(\mathbf{x})}\big[\|x_i-Y\|\big]
\le
\mathbb E_{Y'\sim f(x_i',\mathbf{x}_{-i})}\big[\|x_i-Y'\|\big].
\]
A mechanism \(f\) is \emph{anonymous} if its output is invariant under permutations of the agents.

We evaluate the quality of a facility location by the $L_p$-norm of the agents' distance
vector. For \(p\in[1,+\infty]\), the \emph{$L_p$-norm social cost} of a facility
location \(y\in\mathbb R^2\) on profile \(\mathbf{x}\) is
\[
    \mathrm{SC}_p(\mathbf{x},y)
    =
    \left(
        \sum_{i\in N}\|x_i-y\|^p
    \right)^{\frac{1}{p}},
\]
where for \(p=+\infty\), it is
\(\mathrm{SC}_{+\infty}(\mathbf{x},y)=\max_{i\in N}\|x_i-y\|\). 
The optimal $L_p$-norm social cost on profile
\(\mathbf{x}\) is
\(
    \mathrm{OPT}_p(\mathbf{x})
    =
    \min_{y\in\mathbb R^2}\mathrm{SC}_p(\mathbf{x},y).
\)
The $L_p$-norm social cost achieved by a mechanism $f$ on profile $\mathbf x$ is $\mathrm{SC}_p(\mathbf{x},f(\mathbf{x}))= \mathbb E_{Y\sim f(\mathbf{x})}
    \left[
        \mathrm{SC}_p(\mathbf{x},Y)
    \right].$
    When the context is clear, we write  \(\mathrm{ALG}_p(\mathbf{x})=\mathrm{SC}_p(\mathbf{x},f(\mathbf{x}))\) for simplicity. 
% For a deterministic mechanism \(f\), the $L_p$-norm social cost achieved on
% profile \(\mathbf{x}\) is
% \(\mathrm{ALG}_p(\mathbf{x})=\mathrm{SC}_p(\mathbf{x},f(\mathbf{x}))\). For
% a randomized mechanism $f$, the
% $L_p$-norm social cost on
% profile \(\mathbf{x}\) is the expectation 
% \(
%     \mathrm{ALG}_p(\mathbf{x})
%     =
%     \mathbb E_{Y\sim f(\mathbf{x})}
%     \left[
%         \mathrm{SC}_p(\mathbf{x},Y)
%     \right].
% \)
%In particular, for \(p=+\infty\), the maximum distance is first computed foreach realized facility location \(Y\), and then the expectation is taken over the mechanism's randomness.

The \emph{approximation ratio} of a mechanism $f$ is defined as 
$$\alpha_p(f)=\sup_{\mathbf x}\frac{\mathrm{SC}_p(\mathbf{x},f(\mathbf x))}{\mathrm{OPT}_p(\mathbf{x})}.$$
That is, $\alpha_p(f)$ is the worst-case ratio between the $L_p$-norm social cost achieved by $f$ and the optimal $L_p$-norm social cost, over all instances. 

% For \(p\in[1,+\infty]\), a mechanism is an \(\alpha\)-approximation under the
% $L_p$-norm social cost if, for every profile \(\mathbf{x}\in\left(\mathbb R^2\right)^n\),
% \[
%     \mathrm{ALG}_p(\mathbf{x})
%     \le
%     \alpha\cdot \mathrm{OPT}_p(\mathbf{x}).
% \]
%For profiles with \(\mathrm{OPT}_p(\mathbf{x})=0\), this condition requires \(\mathrm{ALG}_p(\mathbf{x})=0\). 
%\redcomment{[need carefully define $\alpha_p(f)$, as the sup over all instances]}

\section{Coordinate-wise Median Mechanism}\label{sec:cm}

%In the Euclidean plane $\mathbb R^2$, Meir \cite{DBLP:conf/sagt/Meir19} proved that the coordinate-wise median is a \(\sqrt{2}\)-approximation for the standard social cost, i.e., \(p=1\).
Goel and Hann-Caruthers \cite{GoelH23} showed
that, for any $p\ge 1$ the coordinate-wise median
mechanism has the lowest approximation ratio under the $L_p$-norm social cost
among all deterministic, anonymous, and strategyproof mechanisms, and 
 its approximation ratio $\alpha_p(\mathrm{CM})$ is between
\(2^{1-\frac{1}{p}}\) and \(2^{\frac{3}{2}-\frac{2}{p}}\) for any \(p\ge 2\). They further conjecture that this lower bound is tight.

\begin{conjecture}[Goel and Hann-Caruthers \cite{GoelH23}]
     For $\mathbb R^2$ and the $L_p$-norm social cost where $p\ge 2$, the approximation ratio of the coordinate-wise median mechanism is \(2^{1-\frac{1}{p}}\).
\end{conjecture}

In this section, we prove their conjecture firmly. Moreover, we complete the picture by showing that when $1\le p\le 2$ it achieves a tight bound of $\sqrt 2$.
%In this section, we close the gap for \(p\ge 2\) by proving that the lower bound \(2^{1-\frac{1}{p}}\) is tight, and we extend the tight analysis to all \(p\ge 1\).
%For a finite sequence \(s_1,\ldots,s_n\in\mathbb R\), let \(\med\left(s_1,\ldots,s_n\right)\) be the \(\lceil \frac{n}{2}\rceil\)-th smallest element. Equivalently, it is the lower median when \(n\) is even and the usual median when \(n\) is odd.

% \begin{mechanism}[Coordinate-Wise Median]\label{mec:cm}
% Given a profile \(\mathbf{x}=(x_1,\ldots,x_n)\in(\mathbb R^2)^n\), write
% \(x_i=(x_{i,1},x_{i,2})\). The coordinate-wise median mechanism, denoted by
% \(\mathrm{CM}\), outputs
% \[
%     \mathrm{CM}(\mathbf{x})
%     =
%     \left(
%         \med\left(x_{1,1},\ldots,x_{n,1}\right),
%         \med\left(x_{1,2},\ldots,x_{n,2}\right)
%     \right).
% \]
% \end{mechanism}

%\cmmechanism*

\cmtheorem*

% \begin{theorem}\label{thm:cm-tight}
% The coordinate-wise median mechanism is strategyproof. Moreover, its tight
% approximation ratio under the $L_p$-norm social cost is
% \[
%     \alpha_p(\mathrm{CM})
%     =
%     \begin{cases}
%         \sqrt{2}, & 1\le p\le 2,\\[3pt]
%         2^{1-\frac{1}{p}}, & 2\le p<+\infty,\\[3pt]
%         2, & p=+\infty.
%     \end{cases}
% \]
% \end{theorem}
The proof separates the regimes \(1\le p\le2\) and \(p\ge2\). For \(1\le p\le2\), we combine a standard one-dimensional median inequality (Lemma~\ref{lem:one-dimensional-median}) with norm comparisons between \(L_p\) and Euclidean norms in the plane, yielding the tight factor \(\sqrt 2\) in Lemma~\ref{lem:cm-small-p}. For \(p\ge2\), we pair agents with opposite coordinate signs and use a geometric inequality for vectors with nonpositive inner product together with H\"older’s inequality to obtain the tight factor \(2^{1-\frac1p}\) in Lemma~\ref{lem:cm-large-p}.

%We first record a standard one-dimensional inequality for the median.

\begin{lemma}\label{lem:one-dimensional-median}
Let \(s_1,\ldots,s_n\in\mathbb R\), and let
\(m=\med\left(s_1,\ldots,s_n\right)\). Then, for every \(t\in\mathbb R\) and
every \(p\ge 1\),
\begin{align}
    \sum_i |s_i-m|^p
    \le
    2^{p-1}\sum_i |s_i-t|^p .
    \label{eq:one-dimensional-median}
\end{align}
\end{lemma}

\begin{proof}
By translating the line, assume without loss of generality that \(m=0\).  If the number $n$ of points is odd, we
duplicate the multiset. The median of the duplicated multiset is still
\(0\), and proving \eqref{eq:one-dimensional-median} for the duplicated
instance implies the original inequality, since both sides are doubled.
Thus, assume that the total number $n$ of points is even. 

We
first explain how to pair these points. Since
\(0\) is a median, at most \(\frac{n}{2}\) points are strictly positive and at most
\(\frac{n}{2}\) points are strictly negative. We assign artificial signs to the zero
points so that exactly \(\frac{n}{2}\) points have sign \(+\) and exactly \(\frac{n}{2}\)
points have sign \(-\): all strictly positive points are assigned \(+\), all
strictly negative points are assigned \(-\), and the zero points are assigned
arbitrarily to make the two sign classes have equal size. 
% This is possible
% because the numbers of strictly positive and strictly negative points are
% both at most \(\frac{n}{2}\). 
Now pair each \(+\)-signed point with one \(-\)-signed
point. 
% For every resulting pair \((u,v)\), we have \(uv\le 0\), since a
% nonzero point has its actual sign and a point whose artificial sign may be
% chosen freely is zero.

Fix one such pair \((u,v)\). Since \(uv\le 0\), the two points lie on
opposite sides of the origin, possibly with one of them equal to zero. Hence
\(|u|+|v|=|u-v|\). Therefore,
\[
    \left(|u|^p+|v|^p\right)^{\frac{1}{p}}
    \le
    |u|+|v|
    =
    |u-v| = |u-t-(v-t)|
    \le
    |u-t|+|v-t|.
\]

Let \(A=|u-t|\) and \(B=|v-t|\). If \(p=1\), the following inequality is an
equality. For \(p>1\), let \(q=\frac{p}{p-1}\). By H\"older's inequality,
\[
    A+B
    =
    A\cdot 1+B\cdot 1
    \le
    \left(A^p+B^p\right)^{\frac{1}{p}}
    \left(1^q+1^q\right)^{\frac{1}{q}}
    =
    2^{1-\frac{1}{p}}
    \left(A^p+B^p\right)^{\frac{1}{p}}.
\]

Combining the two inequalities and taking \(p\)-th powers, we obtain
\[
    |u|^p+|v|^p
    \le
    2^{p-1}\left(|u-t|^p+|v-t|^p\right).
\]
Summing this inequality over all pairs gives
\eqref{eq:one-dimensional-median}.
\end{proof}

\begin{lemma}\label{lem:cm-small-p}
For every \(1\le p\le 2\), the coordinate-wise median mechanism is a tight
\(\sqrt{2}\)-approximation under the $L_p$-norm social cost. 
\end{lemma}

\begin{proof}
Let \(m=\mathrm{CM}(\mathbf{x})=(m_1,m_2)\in\mathbb R^2\), and let
\(z=(z_1,z_2)\in\mathbb R^2\) be arbitrary. We first recall two elementary
norm inequalities in \(\mathbb R^2\). For every \((u,v)\in\mathbb R^2\) and
\(1\le p\le 2\),
\begin{align}
    \left(u^2+v^2\right)^{\frac{p}{2}}
    &\le
    |u|^p+|v|^p,
    \label{eq:l2-to-lp-small-p}
    \\
    |u|^p+|v|^p
    &\le
    2^{1-\frac{p}{2}}
    \left(u^2+v^2\right)^{\frac{p}{2}}.
    \label{eq:lp-to-l2-small-p}
\end{align}
Indeed, \eqref{eq:l2-to-lp-small-p} is the standard monotonicity of
\(L_p\)-norms in dimension two: when \(p\le 2\),
\[
    \left(u^2+v^2\right)^{\frac{1}{2}}
    \le
    \left(|u|^p+|v|^p\right)^{\frac{1}{p}},
\]
and raising both sides to the \(p\)-th power gives
\eqref{eq:l2-to-lp-small-p}. For \eqref{eq:lp-to-l2-small-p}, let
\(r=\frac{p}{2}\le 1\), \(A=u^2\), and \(B=v^2\). Since the function
\(x^r\) is concave on \(\mathbb R_{\ge 0}\), Jensen's inequality gives
\[
    \left(\frac{A+B}{2}\right)^r
    \ge
    \frac{A^r+B^r}{2}.
\]
Multiplying both sides by \(2\), we get
\[
    A^r+B^r
    \le
    2\left(\frac{A+B}{2}\right)^r
    =
    2^{1-r}(A+B)^r.
\]
Substituting back \(r=\frac{p}{2}\), \(A=u^2\), and \(B=v^2\), this becomes
\[
    |u|^p+|v|^p
    =
    A^r+B^r
    \le
    2^{1-\frac{p}{2}}(A+B)^{\frac{p}{2}}
    =
    2^{1-\frac{p}{2}}
    \left(u^2+v^2\right)^{\frac{p}{2}}.
\]

Now apply \eqref{eq:l2-to-lp-small-p} to each agent's Euclidean distance
from the coordinate-wise median \(m\). For each agent \(i\),
\[
    \|x_i-m\|^p
    =
    \left(
        (x_{i,1}-m_1)^2+(x_{i,2}-m_2)^2
    \right)^{\frac{p}{2}}.
\]
Therefore,
\[
\begin{aligned}
    \mathrm{SC}_p(\mathbf{x},m)^p
    &=
    \sum_i
    \left(
        (x_{i,1}-m_1)^2+(x_{i,2}-m_2)^2
    \right)^{\frac{p}{2}} \\
    &\le
    \sum_i |x_{i,1}-m_1|^p
    +
    \sum_i |x_{i,2}-m_2|^p .
\end{aligned}
\]

Since \(m_1\) is the median of the first coordinates
\(x_{1,1},\ldots,x_{n,1}\), Lemma~\ref{lem:one-dimensional-median} with
\(t=z_1\) gives
\[
    \sum_i |x_{i,1}-m_1|^p
    \le
    2^{p-1}\sum_i |x_{i,1}-z_1|^p.
\]
Similarly, since \(m_2\) is the median of the second coordinates,
Lemma~\ref{lem:one-dimensional-median} with \(t=z_2\) gives
\[
    \sum_i |x_{i,2}-m_2|^p
    \le
    2^{p-1}\sum_i |x_{i,2}-z_2|^p.
\]
Combining the last three inequalities,
\[
\begin{aligned}
    \mathrm{SC}_p(\mathbf{x},m)^p
    &\le
    2^{p-1}
    \sum_i
    \left(
        |x_{i,1}-z_1|^p+|x_{i,2}-z_2|^p
    \right).
\end{aligned}
\]
Finally, applying \eqref{eq:lp-to-l2-small-p} to the vector \(x_i-z\) for each
agent \(i\) gives
\[
\begin{aligned}
    \mathrm{SC}_p(\mathbf{x},m)^p
    &\le
    2^{p-1}
    \sum_i
    2^{1-\frac{p}{2}}
    \left(
        (x_{i,1}-z_1)^2+(x_{i,2}-z_2)^2
    \right)^{\frac{p}{2}} \\
    &=
    2^{\frac{p}{2}}
    \sum_i \|x_i-z\|^p
    =
    2^{\frac{p}{2}}\mathrm{SC}_p(\mathbf{x},z)^p.
\end{aligned}
\]
Taking \(z\) to be an optimal facility location gives
\[
    \mathrm{ALG}_p(\mathbf{x})
    =
    \mathrm{SC}_p(\mathbf{x},m)
    \le
    \sqrt{2}\,\mathrm{OPT}_p(\mathbf{x}).
\]

We now show tightness. Consider the instance $n=2k$ with \(k\) agents at
\((1,0)\) and \(k\) agents at \((0,1)\). The
coordinate-wise median is \((0,0)\), and thus
\(\mathrm{ALG}_p(\mathbf{x})^p=2k\). On the other hand, at
\(z=\left(\frac{1}{2},\frac{1}{2}\right)\), every agent has distance
\(\frac{1}{\sqrt{2}}\), so
\[
    \mathrm{OPT}_p(\mathbf{x})^p
    \le
    2k\cdot2^{-\frac{p}{2}}.
\]
Thus the \(p\)-th power of the approximation ratio is at least $2^{\frac{p}{2}}$. Hence no approximation factor below \(\sqrt{2}\) is possible.
\end{proof}

\begin{lemma}\label{lem:cm-large-p}
For every \(2\le p\le+\infty\), the coordinate-wise median mechanism is a
tight \(2^{1-\frac{1}{p}}\)-approximation under the $L_p$-norm social cost. 
\end{lemma}

\begin{proof}
Let \(m=\mathrm{CM}(\mathbf{x})\). By translating the instance, assume
\(m=(0,0)\). Then, for each coordinate, at most half of the agents lie strictly
on either side of the origin. By duplicating the multiset if necessary, we
may assume that the number of agents is even.

Assign artificial signs \(+\) and \(-\) to zero coordinates so that exactly
half of the agents have \(+\) sign and exactly half have \(-\) sign in each
coordinate. This partitions the agents into four signed classes
\((++),(+-),(-+),(--)\). Since the number of \(+\) signs is exactly half in
both coordinates, the number of \((++)\) agents equals the number of \((--)\)
agents, and the number of \((+-)\) agents equals the number of \((-+)\)
agents. Therefore, we can pair every \((++)\) agent with a \((--)\) agent,
and every \((+-)\) agent with a \((-+)\) agent. For every resulting pair
\((a,b)\in(\mathbb R^2)^2\) of agent locations, we have \(a_1b_1\le 0\) and \(a_2b_2\le 0\), and hence
\(a\cdot b\le 0\).

Fix such a pair \((a,b)\). Since \(a\cdot b\le 0\),
\(\|a-b\|^2\ge \|a\|^2+\|b\|^2\). Moreover, since \(p\ge 2\),
\begin{equation}\label{eq:31}
    \left(\|a\|^p+\|b\|^p\right)^{\frac{1}{p}}
    \le
    \left(\|a\|^2+\|b\|^2\right)^{\frac{1}{2}}
    \le
    \|a-b\|.
\end{equation}
For any comparison point \(z\in\mathbb R^2\), the triangle inequality gives
\begin{equation}\label{eq:32}
    \|a-b\|
    =
    \|(a-z)-(b-z)\|
    \le
    \|a-z\|+\|b-z\|.
\end{equation}
It remains to upper bound the sum on the right by an \(L_p\)-type expression.

First consider finite $p$, namely $2\le p<+\infty$.
Let \(A=\|a-z\|\) and \(B=\|b-z\|\). %If \(p=1\), then \(A+B=\left(A^p+B^p\right)^{\frac{1}{p}}\).  If \(p>1\), let
Let \(q=\frac{p}{p-1}\) be the conjugate exponent so that \(\frac{1}{p}+\frac{1}{q}=1\). By H\"older's inequality applied to the two vectors \((A,B)\) and \((1,1)\),
\[
    A+B
    =
    A\cdot 1+B\cdot 1
    \le
    \left(A^p+B^p\right)^{\frac{1}{p}}
    \left(1^q+1^q\right)^{\frac{1}{q}}.
\]
Since \(1^q+1^q=2\) and \(\frac{1}{q}=1-\frac{1}{p}\), this gives
\[
    A+B
    \le
    2^{1-\frac{1}{p}}
    \left(A^p+B^p\right)^{\frac{1}{p}}.
\]
Substituting \(A=\|a-z\|\) and \(B=\|b-z\|\), we obtain
\[
    \|a-b\|
    \le
    \|a-z\|+\|b-z\|
    \le
    2^{1-\frac{1}{p}}
    \left(\|a-z\|^p+\|b-z\|^p\right)^{\frac{1}{p}}.
\]
Combining the last two inequalities and taking \(p\)-th powers, each pair
satisfies
\[
    \|a\|^p+\|b\|^p
    \le
    2^{p-1}\left(\|a-z\|^p+\|b-z\|^p\right).
\]
Summing over all pairs
\((a,b)\in(\mathbb R^2)^2\) of agent locations gives
\[
    \mathrm{SC}_p(\mathbf{x},m)^p
    \le
    2^{p-1}\mathrm{SC}_p(\mathbf{x},z)^p.
\]
Taking \(z\) to be an optimal facility location yields
\(\mathrm{ALG}_p(\mathbf{x})\le
2^{1-\frac{1}{p}}\mathrm{OPT}_p(\mathbf{x})\).

Second consider $p=+\infty$. Goel and Hann-Caruthers \cite{GoelH23} have shown a tight bound of 2. We present our simple alternative proof here. \eqref{eq:31} and \eqref{eq:32} give
\[
    \max\left\{\|a\|,\|b\|\right\}
    \le
    \|a-b\|
\]
and for any comparison point \(z\in\mathbb R^2\),
\[
    \|a-b\|
    \le
    \|a-z\|+\|b-z\|
    \le
    2\max\left\{\|a-z\|,\|b-z\|\right\},
\]
respectively.
Then taking the maximum over all pairs
\((a,b)\in(\mathbb R^2)^2\) of agent locations gives
\(\mathrm{ALG}_{+\infty}(\mathbf{x})\le
2\,\mathrm{OPT}_{+\infty}(\mathbf{x})\).

\smallskip
For tightness, let \(u=(1,1)\), and consider the instance $n=2k$ with \(k\) agents
at \(u\) and \(k\) agents at the origin. The coordinate-wise median is the
origin. When $2\le p<+\infty$, we have
\(\mathrm{ALG}_p(\mathbf{x})^p=k\|u\|^p=k2^{\frac{p}{2}}\). At the midpoint
\(\frac{u}{2}\) which is the optimum, every agent has distance
\(\frac{\|u\|}{2}=\frac{1}{\sqrt{2}}\), and hence
\[
    \mathrm{OPT}_p(\mathbf{x})^p
    =
    2k\cdot 2^{-\frac{p}{2}} = k\cdot 2^{1-\frac{p}{2}}.
\]
%Thus the \(p\)-th power of the approximation ratio is at least $2^{p-1}$.
Hence the approximation
ratio is at least \(2^{1-\frac{1}{p}}\), matching the upper bound.
When $p=+\infty$, we have
\(\mathrm{ALG}_{+\infty}(\mathbf{x})=\sqrt{2}\). At the midpoint
\(\frac{u}{2}\), the maximum distance is \(\frac{1}{\sqrt{2}}\), and thus
\(\mathrm{OPT}_{+\infty}(\mathbf{x})\le \frac{1}{\sqrt{2}}\). The ratio is
therefore at least \(2\), also matching the upper bound.
\end{proof}

\iffalse
It has been known in \cite{GoelH23} that the coordinate-wise median has a tight bound of 2 for  \(p=+\infty\). We present our proof merely for completeness here.

\begin{lemma}\label{lem:cm-infinity}
For \(p=+\infty\), the coordinate-wise median mechanism is a tight
\(2\)-approximation under the maximum cost.\ma{pending remove}
\end{lemma}

\begin{proof}
We use the same pairing construction as in the proof of
Lemma~\ref{lem:cm-large-p}. After translating so that
\(m=\mathrm{CM}(\mathbf{x})=0\), we can pair the agents so that every pair
\((a,b)\) satisfies \(a\cdot b\le 0\). Hence
\(\|a-b\|^2\ge \|a\|^2+\|b\|^2\), and in particular
\[
    \max\left\{\|a\|,\|b\|\right\}
    \le
    \|a-b\|.
\]
For any comparison point \(z\in\mathbb R^2\),
\[
    \|a-b\|
    \le
    \|a-z\|+\|b-z\|
    \le
    2\max\left\{\|a-z\|,\|b-z\|\right\}.
\]
Taking the maximum over all pairs gives
\(\mathrm{ALG}_{+\infty}(\mathbf{x})\le
2\,\mathrm{OPT}_{+\infty}(\mathbf{x})\).

The lower bound follows from the same instance as in
Lemma~\ref{lem:cm-large-p}: take \(k\) agents at \(u=(1,1)\) and \(k\)
agents at the origin. The coordinate-wise median is the origin, so
\(\mathrm{ALG}_{+\infty}(\mathbf{x})=\sqrt{2}\). At the midpoint
\(\frac{u}{2}\), the maximum distance is \(\frac{1}{\sqrt{2}}\), and thus
\(\mathrm{OPT}_{+\infty}(\mathbf{x})\le \frac{1}{\sqrt{2}}\). The ratio is
therefore at least \(2\), matching the upper bound.
\end{proof}
\fi

\begin{proof}[Proof of Theorem~\ref{thm:cm-tight}]
The coordinate-wise median
mechanism is known to be strategyproof in Euclidean spaces (see, e.g.,
\cite{GoelH23,DBLP:conf/sagt/Meir19}), because it applies a one-dimensional
median independently on each dimension. The approximation guarantees and their
tightness follow from Lemmas~\ref{lem:cm-small-p} and \ref{lem:cm-large-p}.
\end{proof}

%Theorem~\ref{thm:cm-tight} closes the gap left by Goel and Hann-Caruthers \cite{GoelH23}: for every \(p\ge 2\), the approximation ratio of coordinate-wise median is exactly \(2^{1-\frac{1}{p}}\), improving the previous upper bound \(2^{\frac{3}{2}-\frac{2}{p}}\). %this identifies the optimal approximation guarantee among deterministic strategyproof mechanisms for the $L_p$-norm social cost.

\section{Uniformly Rotated Coordinate-wise Median}\label{sec:urcm}

Goel and Hann-Caruthers \cite{GoelH23} suggested  that
randomly rotating the coordinate system before applying the coordinate-wise
median may lead to improved guarantees. In this section, we formalize this
idea and analyze the resulting randomized mechanism. We show that this
rotation indeed improves the approximation ratio for \(1\le p<2\), including
the standard case of total cost, while no improvement is possible for
\(p\ge 2\).

For an angle \(\theta\in[0,2\pi)\), let
\(e_\theta=\left(\cos\theta,\sin\theta\right)\) and
\(e_\theta^\perp=\left(-\sin\theta,\cos\theta\right)\) be two perpendicular unit vectors that form a rotated coordinate system, instead of using the standard axes
$(1,0)$ and $(0,1)$. %it uses the axes $e_\theta$ and $e_\theta^\perp$.
For each agent location
\(x_i\in\mathbb R^2\), the dot product \(\left\langle  x_i,e_\theta\right\rangle\) is the coordinate of $x_i$ along the first rotated axis, and \(\left\langle x_i,e_\theta^\perp\right\rangle\) is the coordinate along the second rotated axis.

The URCM mechanism  (Mechanism \ref{mec:urcm}) first samples a random angle \(\Theta\sim\mathrm{Unif}[0,2\pi)\). Let $m_{\Theta,1}$ be the median of all first coordinates of agent locations, and $m_{\Theta,2}$ the median of all second coordinates. The mechanism then returns a random point that takes $(m_{\Theta,1},m_{\Theta,2})$ as its coordinates in the rotated coordinate system, namely, 
\[
    \mathrm{URCM}(\mathbf{x})
    =
    m_{\Theta,1}e_\Theta + m_{\Theta,2}e_\Theta^\perp .
\]

%the two coordinates of \(x\) in the rotated coordinate system are \(\left\langle x,e_\theta\right\rangle\) and \(\left\langle x,e_\theta^\perp\right\rangle\).

\iffalse
\label{mec:urcm}
Given a profile \(\mathbf{x}=(x_1,\ldots,x_n)\in(\mathbb R^2)^n\), sample
\(\Theta\sim\mathrm{Unif}[0,2\pi)\). For any \(x_i\), its coordinates in the rotated basis
\(\{e_\Theta,e_\Theta^\perp\}\) are
\(\bigl(\langle x_i,e_\Theta\rangle,\langle x_i,e_\Theta^\perp\rangle\bigr)\), \bluecomment{where $\langle\cdot,\cdot\rangle$ is the inner product of vectors.}
Let
\[
    m_{\Theta,1}
    =
    \med\left(
        \left\langle x_1,e_\Theta\right\rangle,\ldots,
        \left\langle x_n,e_\Theta\right\rangle
    \right),
    \qquad
    m_{\Theta,2}
    =
    \med\left(
        \langle x_1,e_\Theta^\perp\rangle,\ldots,
        \langle x_n,e_\Theta^\perp\rangle
    \right)
\]
\bluecomment{be the median of the first and the second coordinates, respectively.}
%where \(\med(\cdot)\) still denotes the median of a list of real numbers, with a fixed tie-breaking rule when \(n\) is even. 
The mechanism finally maps the resulting pair of medians back to the point in $\mathbb R^2$: %outputs the random point
\[
    \mathrm{URCM}(\mathbf{x})
    =
    m_{\Theta,1}e_\Theta + m_{\Theta,2}e_\Theta^\perp .
\]
\fi

%\urcmmechanism*

\urcmtheorem*

% \begin{theorem}\label{thm:urcm}
% The uniformly rotated coordinate-wise median mechanism is strategyproof in
% expectation. Moreover, it achieves the following approximation guarantee
% under the $L_p$-norm social cost:
% \[
%     \alpha_p(\mathrm{URCM})
%     =
%     \begin{cases}
%         \displaystyle
%         2\left(
%             \dfrac{
%                 \Gamma\left(\frac{p+1}{2}\right)
%             }{
%                 \sqrt{\pi}\,\Gamma\left(1+\frac{p}{2}\right)
%             }
%         \right)^{\frac{1}{p}},
%         & 1\le p<2,\\[14pt]
%         2^{1-\frac{1}{p}},
%         & 2\le p<+\infty,\\[5pt]
%         2,
%         & p=+\infty.
%     \end{cases}
% \]
% For \(p=1\), the first expression equals \(\frac{4}{\pi}\); as
% \(p\uparrow 2\), it converges to \(\sqrt{2}\). For every \(1\le p<2\), this
% bound is strictly smaller than the deterministic coordinate-wise median ratio
% \(\sqrt{2}\). The bounds for \(2\le p\le+\infty\) are tight.
% \end{theorem}

% Even for the case $p=1$, this bound is not tight, either. Consider a profile with one agent at $(-1, 0)$ and the other at $(1, 0)$. This mechanism achieves a bound of $\frac{1}{2}+\frac{2}{\pi} < \frac{4}{\pi}$. Therefore there is an obvious gap for $1\le p < 2$.

For \(p=1\), the first expression equals \(\frac{4}{\pi}\). As
\(p\) approaches 2 from below, it increases and converges to \(\sqrt{2}\). For every \(1\le p<2\), this bound is strictly smaller than the bound \(\sqrt{2}\) of the deterministic coordinate-wise median, whereas when $p\ge 2$ it cannot beat the deterministic coordinate-wise median in the asymptotic sense. We prove the theorem by a sequence of lemmas. %As we will show later, the analysis for \(2\le p\le+\infty\) is asymptotically tight.
%We first prove the improved upper bound for \(1\le p<2\). In fact, the same argument also applies at \(p=2\), where it gives \(\sqrt{2}\).

\begin{lemma}\label{lem:urcm-small-p}
For every \(1\le p\le 2\), \(\mathrm{URCM}\) is a
\(
    2\left(
        \frac{
            \Gamma\left(\frac{p+1}{2}\right)
        }{
            \sqrt{\pi}\,\Gamma\left(1+\frac{p}{2}\right)
        }
    \right)^{\frac{1}{p}}
\)
-approximation under the $L_p$-norm social cost. 
\end{lemma}

\begin{proof}
Fix a realization \(\theta\) of the random angle, and let \(y_\theta\) be the
corresponding output of \(\mathrm{URCM}\). Recall that
\(\left(e_\theta,e_\theta^\perp\right)\) is an orthonormal basis of
\(\mathbb R^2\). Thus, for every vector \(w\in\mathbb R^2\),
\[
    \|w\|^2
    =
    \left\langle w,e_\theta\right\rangle^2
    +
    \left\langle w,e_\theta^\perp\right\rangle^2.
\]
Since \(1\le p\le 2\), according to \eqref{eq:l2-to-lp-small-p}, the norm inequality
\[
    \left(A^2+B^2\right)^{\frac{p}{2}}
    \le
    |A|^p+|B|^p
\]
holds for all real numbers \(A,B\). Applying this inequality to
\(w=x_i-y_\theta\), we obtain
\[
\begin{aligned}
    \|x_i-y_\theta\|^p
    &=
    \left(
        \left\langle x_i-y_\theta,e_\theta\right\rangle^2
        +
        \left\langle x_i-y_\theta,e_\theta^\perp\right\rangle^2
    \right)^{\frac{p}{2}} \le
    \left|\left\langle x_i-y_\theta,e_\theta\right\rangle\right|^p
    +
    \left|\left\langle x_i-y_\theta,e_\theta^\perp\right\rangle\right|^p.
\end{aligned}
\]
Summing over all agents gives
\[
\begin{aligned}
    \mathrm{SC}_p(\mathbf{x},y_\theta)^p
    &=
    \sum_i \|x_i-y_\theta\|^p  \le
    \sum_i
    \left(
        \left|\left\langle x_i-y_\theta,e_\theta\right\rangle\right|^p
        +
        \left|\left\langle x_i-y_\theta,e_\theta^\perp\right\rangle\right|^p
    \right).
\end{aligned}
\]

Now we use the fact that \(y_\theta\) is the coordinate-wise median in the
rotated coordinate system. By the definition of \(\mathrm{URCM}\),
\(\left\langle y_\theta,e_\theta\right\rangle\) is the median of
\(\left\langle x_1,e_\theta\right\rangle,\ldots,
\left\langle x_n,e_\theta\right\rangle\), and
\(\left\langle y_\theta,e_\theta^\perp\right\rangle\) is the median of
\(\left\langle x_1,e_\theta^\perp\right\rangle,\ldots,
\left\langle x_n,e_\theta^\perp\right\rangle\). Therefore, for any
comparison point \(z\in\mathbb R^2\), Lemma~\ref{lem:one-dimensional-median}
applied to the first rotated coordinate gives
\[
    \sum_i
    \left|\left\langle x_i-y_\theta,e_\theta\right\rangle\right|^p
    \le
    2^{p-1}
    \sum_i
    \left|\left\langle x_i-z,e_\theta\right\rangle\right|^p,
\]
and the same lemma applied to the second rotated coordinate gives
\[
    \sum_i
    \left|\left\langle x_i-y_\theta,e_\theta^\perp\right\rangle\right|^p
    \le
    2^{p-1}
    \sum_i
    \left|\left\langle x_i-z,e_\theta^\perp\right\rangle\right|^p.
\]
Combining these inequalities, we have, for every fixed \(\theta\) and every
\(z\in\mathbb R^2\),
\[
\begin{aligned}
    \mathrm{SC}_p(\mathbf{x},y_\theta)^p
    \le
    2^{p-1}
    \sum_i
    \left(
        \left|\left\langle x_i-z,e_\theta\right\rangle\right|^p
        +
        \left|\left\langle x_i-z,e_\theta^\perp\right\rangle\right|^p
    \right).
\end{aligned}
\]

Let \(o\) be an optimal facility location for the $L_p$-norm social cost, and
set \(z=o\). We next take expectation over the random angle \(\Theta\). Fix
a non-zero vector \(v\in\mathbb R^2\) and write it as \(v=\|v\|\left(\cos\phi,\sin\phi\right)\) for some angle $\phi$. %If \(v=0\), the following identity is trivial.
Then we have
\[
    \left\langle v,e_\Theta\right\rangle
    =
    \|v\|\cos(\Theta-\phi),
\]
and
\[
    \left|\left\langle v,e_\Theta^\perp\right\rangle\right|
    =
    \|v\|\left|\sin(\Theta-\phi)\right|.
\]
Since \(\Theta\) is uniform on \([0,2\pi)\), so is \(\Theta-\phi\) modulo
\(2\pi\). Hence
\[
\begin{aligned}
    &\mathbb E_\Theta
    \left[
        \left|\left\langle v,e_\Theta\right\rangle\right|^p
        +
        \left|\left\langle v,e_\Theta^\perp\right\rangle\right|^p
    \right] =
    \|v\|^p\,
    \mathbb E_\Theta
    \left[
        |\cos\Theta|^p+|\sin\Theta|^p
    \right].
\end{aligned}
\]
By symmetry,
\(\mathbb E_\Theta\left[|\cos\Theta|^p\right]
=
\mathbb E_\Theta\left[|\sin\Theta|^p\right]\). Therefore,
\[
    \mathbb E_\Theta
    \left[
        \left|\left\langle v,e_\Theta\right\rangle\right|^p
        +
        \left|\left\langle v,e_\Theta^\perp\right\rangle\right|^p
    \right]
    =
    c_p\|v\|^p,
\]
where
\[
    c_p
    :=
    2\,\mathbb E_\Theta\left[|\cos\Theta|^p\right].
\]

It remains to compute this expectation. We have
\[
\begin{aligned}
    \mathbb E_\Theta\left[|\cos\Theta|^p\right]
    &=
    \frac{1}{2\pi}\int_0^{2\pi}|\cos\theta|^p\,d\theta  =
    \frac{2}{\pi}\int_0^{\frac{\pi}{2}}\cos^p\theta\,d\theta.
\end{aligned}
\]
Using the standard beta-gamma integral,
\[
B(x,y) = 2\int_0^{\pi/2} \sin^{2x-1}\theta \cos^{2y-1}\theta \,d\theta = \frac{\Gamma(x)\Gamma(y)}{\Gamma(x+y)}.
\]
Set \(2x-1 = 0\) (i.e., \(x = \frac12\)) and \(2y-1 = p\) (i.e., \(y = \frac{p+1}{2}\)). Then
\[
\int_0^{\pi/2} \cos^p\theta \,d\theta = \frac{1}{2}\cdot B\!\left(\frac12,\frac{p+1}{2}\right)
= \frac{1}{2}\cdot \frac{\Gamma\left(\frac12\right)\,\Gamma\!\left(\frac{p+1}{2}\right)}{\Gamma\!\left(\frac12+\frac{p+1}{2}\right)}.
\]
Since \(\Gamma\left(\frac12\right)=\sqrt{\pi}\) and \(\frac12+\frac{p+1}{2} = 1+\frac{p}{2}\), we obtain
\[
\int_0^{\pi/2} \cos^p\theta \,d\theta = \frac{\sqrt{\pi}\,\Gamma\!\left(\frac{p+1}{2}\right)}{2\,\Gamma\!\left(1+\frac{p}{2}\right)}.
\]
Therefore we obtain
\begin{align}
    c_p
    =
    2\cdot
    \frac{
        \Gamma\left(\frac{p+1}{2}\right)
    }{
        \sqrt{\pi}\,\Gamma\left(1+\frac{p}{2}\right)
    } .
    \label{eq:urcm-cp}
\end{align}

Applying the above identity with \(v=x_i-o\), and taking expectation over
\(\Theta\), gives
\[
\begin{aligned}
    \mathbb E_\Theta
    \left[
        \mathrm{SC}_p(\mathbf{x},y_\Theta)^p
    \right]
    &\le
    2^{p-1}
    \sum_i
    \mathbb E_\Theta
    \left[
        \left|\left\langle x_i-o,e_\Theta\right\rangle\right|^p
        +
        \left|\left\langle x_i-o,e_\Theta^\perp\right\rangle\right|^p
    \right] \\
    &=
    2^{p-1}c_p
    \sum_i \|x_i-o\|^p  =
    2^{p-1}c_p\,\mathrm{OPT}_p(\mathbf{x})^p.
\end{aligned}
\]

Finally, we convert the moment bound into the expected social cost bound.
Since \(p\ge 1\), Jensen's inequality, or equivalently the monotonicity of
\(L_p\)-norms on a probability space, gives
\[
    \mathbb E_\Theta
    \left[
        \mathrm{SC}_p(\mathbf{x},y_\Theta)
    \right]
    \le
    \left(
        \mathbb E_\Theta
        \left[
            \mathrm{SC}_p(\mathbf{x},y_\Theta)^p
        \right]
    \right)^{\frac{1}{p}}.
\]
Thus
\[
\begin{aligned}
    \mathrm{ALG}_p(\mathbf{x})
    &=
    \mathbb E_\Theta
    \left[
        \mathrm{SC}_p(\mathbf{x},y_\Theta)
    \right]  \le
    2^{1-\frac{1}{p}}c_p^{1/p}\,
    \mathrm{OPT}_p(\mathbf{x}).
\end{aligned}
\]
Substituting \eqref{eq:urcm-cp}, the approximation factor is
\[
\begin{aligned}
    2^{1-\frac{1}{p}}c_p^{1/p}
    &=
    2^{1-\frac{1}{p}}
    \left(
        2\cdot
        \frac{
            \Gamma\left(\frac{p+1}{2}\right)
        }{
            \sqrt{\pi}\,\Gamma\left(1+\frac{p}{2}\right)
        }
    \right)^{1/p}  = 
    2\left(
        \frac{
            \Gamma\left(\frac{p+1}{2}\right)
        }{
            \sqrt{\pi}\,\Gamma\left(1+\frac{p}{2}\right)
        }
    \right)^{1/p}.
\end{aligned}
\]
This proves the claimed approximation guarantee.
\end{proof}

\begin{lemma}\label{lem:urcm-improvement}
For every \(1\le p<2\),
\[
    2\left(
        \frac{
            \Gamma\left(\frac{p+1}{2}\right)
        }{
            \sqrt{\pi}\,\Gamma\left(1+\frac{p}{2}\right)
        }
    \right)^{\frac{1}{p}}
    <
    \sqrt{2}.
\]
At the endpoints, the expression equals \(\frac{4}{\pi}\) when \(p=1\) and
\(\sqrt{2}\) when \(p=2\).
\end{lemma}

\begin{proof}
Let \(\Theta\sim\mathrm{Unif}[0,2\pi)\), and define the random variable
\(Z=|\cos\Theta|\). From the computation of \(c_p\) in
\eqref{eq:urcm-cp}, we have
\[
    \mathbb E_\Theta[Z^p]
    =
    \mathbb E_\Theta[|\cos\Theta|^p]
    =
    \frac{
        \Gamma\left(\frac{p+1}{2}\right)
    }{
        \sqrt{\pi}\,\Gamma\left(1+\frac{p}{2}\right)
    }.
\]
Therefore the approximation factor in Lemma~\ref{lem:urcm-small-p} can be
written as
\[
    2\left(
        \frac{
            \Gamma\left(\frac{p+1}{2}\right)
        }{
            \sqrt{\pi}\,\Gamma\left(1+\frac{p}{2}\right)
        }
    \right)^{\frac{1}{p}}
    =
    2\left(\mathbb E_\Theta[Z^p]\right)^{\frac{1}{p}}
    =
    2\|Z\|_{L_p}.
\]
Here \(\|Z\|_{L_p}=\left(\mathbb E_\Theta[|Z|^p]\right)^{\frac{1}{p}}\) is
the usual \(L_p\)-norm on the probability space
\(\left([0,2\pi),\mathrm{Unif}\right)\).

We next compare this quantity with its value at \(p=2\). On a probability
space, \(L_p\)-norms are monotone in \(p\): if \(1\le p<q\), then
\(\|Z\|_{L_p}\le \|Z\|_{L_q}\). Moreover, the inequality is strict when
\(Z\) is not almost surely constant. In our case, \(Z=|\cos\Theta|\) is not
almost surely constant, since it takes values near \(0\) and near \(1\) on
sets of positive measure. Hence, for every \(1\le p<2\),
\[
    2\|Z\|_{L_p}
    <
    2\|Z\|_{L_2}.
\]
Finally,
\[
    \|Z\|_{L_2}
    =
    \left(\mathbb E_\Theta\left[\cos^2\Theta\right]\right)^{\frac{1}{2}}
    =
    \left(\mathbb E_\Theta\left[\sin^2\Theta\right]\right)^{\frac{1}{2}}
    =
    \left(\mathbb E_\Theta\left[\frac{\cos^2\Theta+\sin^2\Theta}{2}\right]\right)^{\frac{1}{2}}
    =
    \left(\frac{1}{2}\right)^{\frac{1}{2}},
\]
and therefore \(2\|Z\|_{L_2}=\sqrt{2}\). This proves the strict improvement
over \(\sqrt{2}\) for every \(1\le p<2\).

It remains to compute the endpoint \(p=1\). We have
\[
\begin{aligned}
    2\|Z\|_{L_1}
    &=
    2\mathbb E_\Theta[|\cos\Theta|]  =
    2\cdot \frac{1}{2\pi}\int_0^{2\pi}|\cos\theta|\,d\theta  =
    2\cdot \frac{1}{2\pi}\cdot 4
    =
    \frac{4}{\pi}.
\end{aligned}
\]
At \(p=2\), as computed above, the expression is
\(2\|Z\|_{L_2}=\sqrt{2}\).
\end{proof}

We next show that for \(p\ge 2\), random rotation cannot improve the
bound \(2^{1-\frac{1}{p}}\) of the (deterministic) coordinate-wise median.

\begin{lemma}\label{lem:urcm-large-p}
For every \(2\le p\le +\infty\), the approximation ratio  \(\alpha_p(\mathrm{URCM})\) of the $\mathrm{URCM}$ mechanism  %tight approximation ratio \(2^{1-\frac{1}{p}}\).
lies between $\left[2^{1-\frac{1}{p}}(1-\frac1n)^{\frac1p}, 2^{1-\frac{1}{p}}\right]$, where the bound is asymptotically tight. %For \(p=+\infty\), \(\mathrm{URCM}\) has tight approximation ratio \(2\).
\end{lemma}

\begin{proof}
We first consider $2\le p\le +\infty$.
For any fixed angle \(\theta\), the realized deterministic mechanism is just
coordinate-wise median in the rotated coordinate system
\(\left(e_\theta,e_\theta^\perp\right)\). Since rotations preserve Euclidean
distances and the $L_p$-norm social cost, Lemma~\ref{lem:cm-large-p} implies
that, for every profile \(\mathbf{x}\) and every \(\theta\),
\[
    \mathrm{SC}_p(\mathbf{x},y_\theta)
    \le
    2^{1-\frac{1}{p}}\mathrm{OPT}_p(\mathbf{x}).
\]
Taking expectation over \(\Theta\) gives the same upper bound for
\(\mathrm{URCM}\).

For the lower bound, let \(u=(1,1)\), and consider a ``majority-at-the-origin" instance with \(k\)
agents at \(u\) and \(k+1\) agents at the origin. For every rotation angle
\(\theta\), the origin is a strict majority in both rotated coordinates.
Therefore, \(\mathrm{URCM}\) always outputs the origin. Hence
\(\mathrm{ALG}_p(\mathbf{x})=(k\|u\|^p)^{1/p}=k^{1/p}2^{1/2}\). At the midpoint
\(\frac{u}{2}\), every agent has distance
\(\frac{\|u\|}{2}=\frac{1}{\sqrt{2}}\), so
\(
    \mathrm{OPT}_p(\mathbf{x})
    \le
    ((2k+1)2^{-p/2})^{1/p}=n^{1/p}2^{-1/2}.
\)
Thus the approximation ratio is at least
\[
    \frac{k^{\frac1p}2^{\frac12}}{n^{\frac1p}2^{-\frac12}} = 2^{1-\frac{1}{p}}(1-\frac1n)^{\frac1p},
\]
which converges to \(2^{1-1/p}\) as \(n\to\infty\). 
%Therefore the approximation ratio is at least \(2^{1-\frac{1}{p}}(1-\frac1n)^{\frac1p}\). %asymptotically matching the upper bound.
\iffalse
Then we turn to $p=+\infty$. For each fixed angle \(\theta\), the realized mechanism is coordinate-wise
median in a rotated coordinate system. \redcomment{By Lemma~\ref{lem:cm-infinity}, }for
every profile \(\mathbf{x}\),
\[
    \mathrm{SC}_{+\infty}(\mathbf{x},y_\theta)
    \le
    2\,\mathrm{OPT}_{+\infty}(\mathbf{x}).
\]
Taking expectation over \(\Theta\) gives the upper bound for
\(\mathrm{URCM}\).

The tightness \bluecomment{is given by the same majority-at-the-origin instance}. %take the same instance as in Lemma~\ref{lem:urcm-large-p}: \(k\) agents are at \(u=(1,1)\), and \(k+1\) agents are at the origin. 
For
every rotation angle, the origin is a strict majority in both rotated
coordinates, so \(\mathrm{URCM}\) always outputs the origin. Therefore
\(\mathrm{ALG}_{+\infty}(\mathbf{x})=\sqrt{2}\). At the midpoint
\(\frac{u}{2}\), the maximum distance is \(\frac{1}{\sqrt{2}}\), and hence
\(\mathrm{OPT}_{+\infty}(\mathbf{x})\le\frac{1}{\sqrt{2}}\). The ratio is at
least \(2\), matching the upper bound.
\fi
\end{proof}

Also we show a lower bound of URCM by the following cases.

\begin{lemma}\label{lem:urcm-small-p-lower}
For every \(1\le p\le 2\),
\[
    \alpha_p(\mathrm{URCM})
    \ge
    2^{1-\frac{1}{p}}\frac{4}{\pi}
    \int_0^{\frac{\pi}{4}}
    \left(\cos^p t+\sin^p t\right)^{\frac{1}{p}}\,dt.
\]
In particular, this lower bound equals \(\frac{4}{\pi}\) at \(p=1\) and
\(\sqrt{2}\) at \(p=2\).
\end{lemma}

\begin{proof}
Fix \(1\le p\le 2\). We construct a sequence of instances whose
approximation ratio converges to the claimed value. Let \(M\) be a positive
integer and let \(k=M^3\). Consider the profile with \(k\) agents at $A=(1,0)$,
\(k\) agents at $B=(0,1)$,
and one agent at $C=(-M,-M)$.
The number of agents is \(2k+1\), so in each coordinate the median is the
\((k+1)\)-st order statistic.

By the \(\frac{\pi}{2}\)-periodicity of the rotated coordinate-wise median
on odd-size profiles, it is enough to average over
\(\theta\sim\mathrm{Unif}[0,\frac{\pi}{2})\). Write
\(c=\cos\theta\) and \(s=\sin\theta\). After rotation by \(\theta\), the
three locations become
$A'=(c,s)$, $B'=(-s,c)$, and $C'=(-M(c-s),-M(c+s))$.
Let \(m_\theta\) be the coordinate-wise median of the rotated profile.

The only coordinate in which the outlier can affect the median, except for a
small set of angles, is the first coordinate. Define
\[
    W_M
    =
    \left\{
        \theta\in\left[0,\frac{\pi}{2}\right):
        -M(\cos\theta-\sin\theta)\in(-\sin\theta,\cos\theta)
    \right\}.
\]
We first note that
\[
    |W_M|\le \frac{2}{M}.
\]
Indeed, for \(\theta\in(0,\frac{\pi}{4}]\), the condition
\(-M(c-s)>-s\) is equivalent to
\[
    \tan\theta>\frac{M}{M+1}.
\]
Hence
\[
    W_M\cap\left(0,\frac{\pi}{4}\right]
    \subseteq
    \left(
        \arctan\frac{M}{M+1},
        \frac{\pi}{4}
    \right],
\]
whose length is at most \(\frac{1}{M}\). Similarly, for
\(\theta\in[\frac{\pi}{4},\frac{\pi}{2})\), the condition
\(-M(c-s)<c\) is equivalent to
\[
    \tan\theta<\frac{M+1}{M},
\]
and therefore
\[
    W_M\cap\left[\frac{\pi}{4},\frac{\pi}{2}\right)
    \subseteq
    \left[
        \frac{\pi}{4},
        \arctan\left(1+\frac{1}{M}\right)
    \right),
\]
whose length is also at most \(\frac{1}{M}\). Thus
\(|W_M|\le\frac{2}{M}\).

For every \(\theta\notin W_M\), the median has a simple form. If
\(\theta\in[0,\frac{\pi}{4}]\setminus W_M\), then $m_\theta=(-s,s)$; and
if \(\theta\in[\frac{\pi}{4},\frac{\pi}{2})\setminus W_M\), then $m_\theta=(c,c)$.
In both cases, the two cluster distances satisfy
\[
    \|A'-m_\theta\|=\cos\theta+\sin\theta,
    \qquad
    \|B'-m_\theta\|=|\cos\theta-\sin\theta|,
\]
up to interchanging the two distances. Hence, for every
\(\theta\notin W_M\),
\[
\begin{aligned}
    \mathrm{SC}_p(P',m_\theta)
    &\ge
    \left(
        k(\cos\theta+\sin\theta)^p
        +
        k|\cos\theta-\sin\theta|^p
    \right)^{\frac{1}{p}}  =
    k^{\frac{1}{p}}
    \left(
        (\cos\theta+\sin\theta)^p
        +
        |\cos\theta-\sin\theta|^p
    \right)^{\frac{1}{p}},
\end{aligned}
\]
where \(P'\) is the rotated profile. Since rotations preserve all Euclidean
distances, this is also a lower bound on the realized cost of
\(\mathrm{URCM}\) on the original profile.

Therefore
\[
\begin{aligned}
    \mathrm{ALG}_p(P)
    &=
    \mathbb E_\theta\left[\mathrm{SC}_p(P',m_\theta)\right] \ge
    \frac{2}{\pi}k^{\frac{1}{p}}
    \int_{\left[0,\frac{\pi}{2}\right)\setminus W_M}
    \left(
        (\cos\theta+\sin\theta)^p
        +
        |\cos\theta-\sin\theta|^p
    \right)^{\frac{1}{p}}\,d\theta.
\end{aligned}
\]
The integrand is bounded by an absolute constant for \(1\le p\le 2\), while
\(|W_M|\le\frac{2}{M}\). Hence
\[
\begin{aligned}
    \mathrm{ALG}_p(P)
    &\ge
    \frac{2}{\pi}k^{\frac{1}{p}}
    \int_0^{\frac{\pi}{2}}
    \left(
        (\cos\theta+\sin\theta)^p
        +
        |\cos\theta-\sin\theta|^p
    \right)^{\frac{1}{p}}\,d\theta
    -
    o\left(k^{\frac{1}{p}}\right).
\end{aligned}
\]
By symmetry and the substitution \(t=\frac{\pi}{4}-\theta\) on
\([0,\frac{\pi}{4}]\),
\[
\begin{aligned}
    &\int_0^{\frac{\pi}{2}}
    \left(
        (\cos\theta+\sin\theta)^p
        +
        |\cos\theta-\sin\theta|^p
    \right)^{\frac{1}{p}}\,d\theta  =
    2\sqrt{2}
    \int_0^{\frac{\pi}{4}}
    \left(\cos^p t+\sin^p t\right)^{\frac{1}{p}}\,dt.
\end{aligned}
\]
It follows that
\[
    \mathrm{ALG}_p(P)
    \ge
    \frac{4\sqrt{2}}{\pi}
    k^{\frac{1}{p}}
    \int_0^{\frac{\pi}{4}}
    \left(\cos^p t+\sin^p t\right)^{\frac{1}{p}}\,dt
    -
    o\left(k^{\frac{1}{p}}\right).
\]

It remains to upper bound the optimum. Consider the point
\(z=(\frac{1}{2},\frac{1}{2})\). The distance from \(z\) to both \(A\) and
\(B\) is \(\frac{1}{\sqrt{2}}\), while the distance from \(z\) to \(C\) is
\(\sqrt{2}(M+\frac{1}{2})\). Therefore
\[
\begin{aligned}
    \mathrm{OPT}_p(P)^p
    &\le
    2k\cdot 2^{-\frac{p}{2}}
    +
    \left(\sqrt{2}\left(M+\frac{1}{2}\right)\right)^p  =
    k\,2^{1-\frac{p}{2}}+O(M^p).
\end{aligned}
\]
Since \(k=M^3\) and \(p\le 2\), we have \(M^p=o(k)\). Hence
\[
    \mathrm{OPT}_p(P)
    \le
    k^{\frac{1}{p}}2^{\frac{1}{p}-\frac{1}{2}}(1+o(1)).
\]
Combining the lower bound on \(\mathrm{ALG}_p(P)\) with this upper bound on
\(\mathrm{OPT}_p(P)\), and then taking \(M\to+\infty\), gives
\[
\begin{aligned}
    \alpha_p(\mathrm{URCM})
    &\ge
    \frac{
        \frac{4\sqrt{2}}{\pi}
        \int_0^{\frac{\pi}{4}}
        \left(\cos^p t+\sin^p t\right)^{\frac{1}{p}}\,dt
    }{
        2^{\frac{1}{p}-\frac{1}{2}}
    }  =
    2^{1-\frac{1}{p}}\frac{4}{\pi}
    \int_0^{\frac{\pi}{4}}
    \left(\cos^p t+\sin^p t\right)^{\frac{1}{p}}\,dt.
\end{aligned}
\]

At \(p=1\),
\[
    \int_0^{\frac{\pi}{4}}(\cos t+\sin t)\,dt=1,
\]
so the lower bound is \(\frac{4}{\pi}\). At \(p=2\),
\[
    \int_0^{\frac{\pi}{4}}
    \left(\cos^2 t+\sin^2 t\right)^{\frac{1}{2}}\,dt
    =
    \frac{\pi}{4},
\]
so the lower bound is \(\sqrt{2}\).
\end{proof}

% \begin{lemma}\label{lem:urcm-infinity}
% For \(p=+\infty\), \(\mathrm{URCM}\) has tight approximation ratio \(2\). \redcomment{[combine with Lem~\ref{lem:urcm-large-p} together?] \bluecomment{I agree.}}
% \end{lemma}

% \begin{proof}
% For each fixed angle \(\theta\), the realized mechanism is coordinate-wise
% median in a rotated coordinate system. By Lemma~\ref{lem:cm-infinity}, for
% every profile \(\mathbf{x}\),
% \[
%     \mathrm{SC}_{+\infty}(\mathbf{x},y_\theta)
%     \le
%     2\,\mathrm{OPT}_{+\infty}(\mathbf{x}).
% \]
% Taking expectation over \(\Theta\) gives the upper bound for
% \(\mathrm{URCM}\).

% For tightness, take the same instance as in Lemma~\ref{lem:urcm-large-p}:
% \(k\) agents are at \(u=(1,1)\), and \(k+1\) agents are at the origin. For
% every rotation angle, the origin is a strict majority in both rotated
% coordinates, so \(\mathrm{URCM}\) always outputs the origin. Therefore
% \(\mathrm{ALG}_{+\infty}(\mathbf{x})=\sqrt{2}\). At the midpoint
% \(\frac{u}{2}\), the maximum distance is \(\frac{1}{\sqrt{2}}\), and hence
% \(\mathrm{OPT}_{+\infty}(\mathbf{x})\le\frac{1}{\sqrt{2}}\). The ratio is at
% least \(2\), matching the upper bound.
% \end{proof}

\begin{proof}[Proof of Theorem~\ref{thm:urcm}]
Strategyproofness in expectation follows immediately from the fact that, for
every fixed rotation angle \(\theta\), the realized mechanism is a
coordinate-wise median mechanism in a rotated coordinate system and is
therefore strategyproof. A distribution over deterministic strategyproof 
mechanisms is strategyproof in expectation.

The approximation ratio for \(1\le p\le 2\) follows from
Lemma~\ref{lem:urcm-small-p} and \ref{lem:urcm-small-p-lower}. The special values at \(p=1\) and \(p=2\), as
well as the strict improvement over \(\sqrt{2}\) for \(1\le p<2\), follow
from Lemma~\ref{lem:urcm-improvement}. The guarantees for
\(2\le p\le +\infty\) follow from
Lemma~\ref{lem:urcm-large-p}.
\end{proof}

Theorem~\ref{thm:urcm} %implements the random-rotation idea suggested by Goel and Hann-Caruthers \cite{GoelH23}. It
gives the first randomized strategyproof mechanism that improves the
\(\sqrt{2}\) guarantee of deterministic coordinate-wise median for the \(L_1\) social
cost, achieving a \(\frac{4}{\pi}\)-approximation, and more generally, the
improvement holds  for all \(1\le p<2\). For \(p\ge 2\), the
majority-at-the-origin instance shows that random rotation cannot improve
the deterministic guarantee \(2^{1-\frac{1}{p}}\).  %and the same phenomenon gives the tight factor \(2\) for \(p=+\infty\).

We now explain the reasons. When \(1 \le p < 2\), the \(p\)-th powers of the two
coordinate projections of a vector depend on the choice of coordinate axes.
A poorly chosen fixed coordinate system may overestimate the contribution of
some directions. URCM avoids committing to any one such coordinate system by
averaging uniformly over all rotations. Precisely, for a fixed vector \(v\), after a random rotation, its two projected lengths
are proportional to \(|\cos \Theta|\) and \(|\sin \Theta|\). Hence the
expected contribution of \(v\) is governed by the average value of
\(|\cos \Theta|^p\) over a uniformly random angle \(\Theta\). This average is obtained by evaluating the integral of \(|\cos \theta|^p\) over the
circle using the standard beta-gamma identity, and  this is where the
Gamma-function expression in the approximation ratio comes from.

The distinction between \(p<2\) and \(p=2\) is  reflected in the basic
norm inequality
\[
    \left(A^2+B^2\right)^{\frac{p}{2}}
    \le
    |A|^p+|B|^p,
\]
which strictly holds for \(1 \le p < 2\) and non-zero reals \(A\) and \(B\). At \(p=2\), it becomes an equality.
Equivalently, for a rotated vector whose projected lengths are proportional
to \(|\cos \Theta|\) and \(|\sin \Theta|\), we have
\(
    \cos^2 \Theta+\sin^2 \Theta=1.
\)
Thus, at \(p=2\), the sum of the squared projected lengths is exactly the
squared Euclidean length of the vector for every rotation. The averaging
advantage therefore disappears.
For \(p \ge 2\), every fixed rotated coordinate-wise median already satisfies
the deterministic bound \(2^{1-1/p}\). Consequently, averaging over rotations
cannot improve the worst-case guarantee in this regime.

\section{Centroid Random Dictatorship}\label{sec:crd}

Feldman and Wilf  \cite{feldman2013strategyproof} first introduced a simple randomized strategyproof mechanism on the real line that mixes the centroid and the random dictatorship, with half probability for each.
Later, Tang et al. \cite{tang2020characterization} applied it to multi-dimensional Euclidean spaces.
They proved the strategyproofness of this mechanism and showed that it achieves the tight ratio \(2-\frac{1}{n}\) for the max-cost objective. The mechanism was not given a specific name there; we call it
\emph{Centroid Random Dictatorship}, abbreviated as \(\mathrm{CRD}\).
Given a profile \(\mathbf{x}\), \(\mathrm{CRD}\) returns 
the centroid \(\bar x=\frac{1}{n}\sum_i x_i\)  with probability
\(\frac{1}{2}\), 
and returns each agent location \(x_j\) with probability
\(\frac{1}{2n}\).

% \begin{mechanism}[Centroid Random Dictatorship]\label{mec:crd}
% Given a profile \(\mathbf{x}=(x_1,\ldots,x_n)\in(\mathbb R^2)^n\), let
% \(c=\frac{1}{n}\sum_i x_i\) be the centroid. The centroid random dictatorship
% mechanism, denoted by \(\mathrm{CRD}\), returns \(c\) with probability
% \(\frac{1}{2}\), and returns each agent location \(x_j\) with probability
% \(\frac{1}{2n}\).
% \end{mechanism}

%\crdmechanism*

%The following theorem gives the approximation ratio of \(\mathrm{CRD}\) for all \(p\ge 1\).

\crdtheorem*

When $p=1$, the above upper bound is equal to $2-\frac2n$. As $p$ approaches $2$ from below, it decreases to  $\frac{1+\sqrt2}{2}\approx 1.2$; as $p$ approaches infinity it increases to $2-\frac1n$.  %For a large enough number of agents $n$, the approximation ratio $\alpha_p(\mathrm{CRD})$ is smaller than  $\alpha_p(\mathrm{URCM})$ when $p$ is greater than approximately 1.6.

% \begin{theorem}\label{thm:crd}
% The centroid random dictatorship mechanism is strategyproof in expectation.
% Moreover, under the $L_p$-norm social cost, it satisfies
% \[
%     \alpha_p(\mathrm{CRD})
%     \le
%     \begin{cases}
%         \displaystyle
%         \frac{1}{2}q_n^{\frac{2}{p}-1}
%         \left(
%             2^{\frac{2}{p}-1}+2^{\frac{1}{p}}
%         \right),
%         & 1\le p\le 2,\\[12pt]
%         \displaystyle
%         \frac{1}{2}
%         \left(
%             2^{1-\frac{2}{p}}q_n^{1-\frac{2}{p}}
%             +
%             2^{1-\frac{1}{p}}
%         \right),
%         & 2\le p<+\infty,\\[12pt]
%         \displaystyle
%         2-\frac{1}{n},
%         & p=+\infty.
%     \end{cases}
% \]
% In particular, at the endpoints,
% \[
%     \alpha_1(\mathrm{CRD})\le 2-\frac{2}{n},\qquad
%     \alpha_2(\mathrm{CRD})\le \frac{1+\sqrt{2}}{2},\qquad
%     \alpha_{+\infty}(\mathrm{CRD})\le 2-\frac{1}{n}.
% \]
% These endpoint bounds are tight for every \(n\). Moreover, for \(n=2\), the
% above bound is tight for every \(1\le p\le+\infty\).
% \end{theorem}

We prove the theorem through a norm-operator analysis. 
For a vector \(X=(x_1,\ldots,x_n)\in\left(\mathbb R^2\right)^n\), define the \emph{normalized \(L_p\)-norm} as
\[
    \|X\|_p
    =
    \left(
        \frac{1}{n}\sum_i \|x_i\|^p
    \right)^{\frac{1}{p}},
    \qquad
    \|X\|_\infty=\max_i\|x_i\|.
\]
Let \(\bar x=\frac{1}{n}\sum_i x_i\). Define two linear operators \(T:(\mathbb R^2)^n\to(\mathbb R^2)^n\) and \(S:(\mathbb R^2)^n\to(\mathbb R^2)^{n^2}\) by
\[
    (TX)_i=x_i-\bar x,
    \qquad
    (SX)_{i,j}=x_i-x_j.
\]
We measure \(TX\) using the same normalized \(L_p\)-norm defined above. For \(SX\), the normalized \(L_p\)-norm is
\[
    \|SX\|_p
    =
    \left(
        \frac{1}{n^2}\sum_{i,j}\|x_i-x_j\|^p
    \right)^{\frac{1}{p}},
    \qquad
    \|SX\|_\infty=\max_{i,j}\|x_i-x_j\|.
\]
Here, \(T\) is the centering operator, which is used to control the expected cost of the centroid part of CRD. The operator \(S\) is the pairwise-difference operator, which is used to control the expected cost of the random-dictatorship part of CRD.
%\bluecomment{The expected cost of CRD naturally decomposes into two terms: the cost of the centroid part and the cost of the random-dictatorship part. The first term is exactly controlled by , which maps each agent location to its deviation from the centroid. The second term is controlled by , which records all distances between pairs of agents}

\begin{lemma}\label{lem:crd-operator}
For \(1\le p\le 2\),
\[
    \|TX\|_p
    \le
    (2q_n)^{\frac{2}{p}-1}\|X\|_p,
    \qquad
    \|SX\|_p
    \le
    2^{\frac{1}{p}}q_n^{\frac{2}{p}-1}\|X\|_p.
\]
For \(2\le p\le+\infty\),
\[
    \|TX\|_p
    \le
    (2q_n)^{1-\frac{2}{p}}\|X\|_p,
    \qquad
    \|SX\|_p
    \le
    2^{1-\frac{1}{p}}\|X\|_p.
\]
% where the first inequality is interpreted as \(\|TX\|_\infty\le 2q_n\|X\|_\infty\) when \(p=+\infty\)\redcomment{[why need the last sentence?] \bluecomment{we do not need this, and can delete it}}.
\end{lemma}

\begin{proof}
% Recall that \(X=(x_1,\ldots,x_n)\), \(\bar x=\frac{1}{n}\sum_i x_i\),
% \((TX)_i=x_i-\bar x\), and \((SX)_{i,j}=x_i-x_j\). Also recall that the
% norms are normalized:
% \[
%     \|X\|_p
%     =
%     \left(
%         \frac{1}{n}\sum_i\|x_i\|^p
%     \right)^{\frac{1}{p}},
%     \qquad
%     \|SX\|_p
%     =
%     \left(
%         \frac{1}{n^2}\sum_{i,j}\|x_i-x_j\|^p
%     \right)^{\frac{1}{p}}.
% \]
We first prove endpoint bounds for \(T\) and \(S\), and then interpolate
between these endpoints.

We begin with \(T\). For \(p=1\), using
\(x_i-\bar x=\frac{1}{n}\sum_j(x_i-x_j)\), the triangle inequality gives
\[
\begin{aligned}
    \|TX\|_1
    &=
    \frac{1}{n}\sum_i \|x_i-\bar x\|  =
    \frac{1}{n}\sum_i
    \|
        \frac{1}{n}\sum_j (x_i-x_j)
    \| \le
    \frac{1}{n^2}\sum_{i,j}\|x_i-x_j\|.
\end{aligned}
\]
The terms with \(i=j\) are zero. For \(i\ne j\), by the triangle inequality,
\(\|x_i-x_j\|\le \|x_i\|+\|x_j\|\). Hence
\[
\begin{aligned}
    \frac{1}{n^2}\sum_{i,j}\|x_i-x_j\|
    &=
    \frac{1}{n^2}\sum_{i\ne j}\|x_i-x_j\| \le
    \frac{1}{n^2}\sum_{i\ne j}\left(\|x_i\|+\|x_j\|\right).
\end{aligned}
\]
Each \(\|x_i\|\) appears \(n-1\) times in the first part and \(n-1\) times
in the second part, so
\[
    \frac{1}{n^2}\sum_{i\ne j}\left(\|x_i\|+\|x_j\|\right)
    =
    \frac{2(n-1)}{n^2}\sum_i\|x_i\|
    =
    2\left(1-\frac{1}{n}\right)\|X\|_1.
\]
Thus
\[
    \|TX\|_1\le 2q_n\|X\|_1.
\]

For \(p=2\), centering by \(\bar x\) is an orthogonal projection. More
explicitly,
\[
\begin{aligned}
    \sum_i\|x_i-\bar x\|^2
    &=
    \sum_i\left(\|x_i\|^2-2\langle x_i,\bar x\rangle+\|\bar x\|^2\right) =
    \sum_i\|x_i\|^2
    -2\left\langle \sum_i x_i,\bar x\right\rangle
    +n\|\bar x\|^2.
\end{aligned}
\]
Since \(\sum_i x_i=n\bar x\), this becomes
\begin{equation}\label{eq:252}
    \sum_i\|x_i-\bar x\|^2
    =
    \sum_i\|x_i\|^2-n\|\bar x\|^2
    \le
    \sum_i\|x_i\|^2.
\end{equation}
Dividing by \(n\), we get
\[
    \|TX\|_2\le \|X\|_2.
\]

For \(p=+\infty\), for each \(i\), we have
% \[
% \begin{aligned}
%     \|x_i-\bar x\| &= \frac{1}{n}\sum_{j}(x_i-x_j) = \frac{1}{n}\sum_{j\ne i}(x_i-x_j)
%     \le
%     \frac{1}{n}\sum_{j\ne i}\|x_i-x_j\|  \le
%     \frac{1}{n}\sum_{j\ne i}\left(\|x_i\|+\|x_j\|\right)  \\
%     &\le
%     \frac{1}{n}\sum_{j\ne i}2\|X\|_\infty
%     =
%     2q_n\|X\|_\infty.
% \end{aligned}
% \]
\[
\begin{aligned}
    \|x_i-\bar x\|
    &=
    \frac{1}{n}\|
        \sum_j(x_i-x_j)
    \|
    =
    \frac{1}{n}\|
        \sum_{j\ne i}(x_i-x_j)
    \|  \le
    \frac{1}{n}\sum_{j\ne i}\|x_i-x_j\| \\
    &\le \frac{1}{n}\sum_{j\ne i}\left(\|x_i\|+\|x_j\|\right)  \le
    \frac{1}{n}\sum_{j\ne i}2\|X\|_\infty
    =
    2q_n\|X\|_\infty.
\end{aligned}
\]
Taking the maximum over \(i\) gives
\[
    \|TX\|_\infty\le 2q_n\|X\|_\infty.
\]

We next prove the endpoint bounds for \(S\). For \(p=1\),
\[
\begin{aligned}
    \|SX\|_1
    &=
    \frac{1}{n^2}\sum_{i,j}\|x_i-x_j\|  =
    \frac{1}{n^2}\sum_{i\ne j}\|x_i-x_j\|  \le
    \frac{1}{n^2}\sum_{i\ne j}\left(\|x_i\|+\|x_j\|\right)=
    2q_n\|X\|_1.
\end{aligned}
\]

For \(p=2\), we use the standard pairwise-variance identity, and we get
\[
\begin{aligned}
    \sum_{i,j}\|x_i-x_j\|^2
    &=
    \sum_{i,j}\|x_i\|^2
    +
    \sum_{i,j}\|x_j\|^2
    -
    2\sum_{i,j}\langle x_i,x_j\rangle.
\end{aligned}
\]
The first two sums are easy to simplify. In
\(\sum_{i,j}\|x_i\|^2\), for each fixed \(i\), the term \(\|x_i\|^2\)
appears once for every \(j\), hence \(n\) times. Therefore
\[
    \sum_{i,j}\|x_i\|^2
    =
    n\sum_i\|x_i\|^2.
\]
Similarly,
\[
    \sum_{i,j}\|x_j\|^2
    =
    n\sum_j\|x_j\|^2
    =
    n\sum_i\|x_i\|^2.
\]
For the inner-product term, by bilinearity of the inner product,
\[
    \sum_{i,j}\langle x_i,x_j\rangle
    =
    \left\langle \sum_i x_i,\sum_j x_j\right\rangle
    =
    \left\|\sum_i x_i\right\|^2.
\]
Combining these three identities gives
\[
    \sum_{i,j}\|x_i-x_j\|^2
    =
    2n\sum_i\|x_i\|^2
    -
    2\left\|\sum_i x_i\right\|^2.
\]
Since \(\bar x=\frac{1}{n}\sum_i x_i\), we have
\(\sum_i x_i=n\bar x\), and hence
\[
    \sum_{i,j}\|x_i-x_j\|^2
    =
    2n\sum_i\|x_i\|^2
    -
    2n^2\|\bar x\|^2.
\]
Now rewrite the right-hand side in terms of the centered vectors \(x_i-\bar x\). 
% Expanding,
% \[
% \begin{aligned}
%     \sum_i\|x_i-\bar x\|^2
%     &=
%     \sum_i
%     \left(
%         \|x_i\|^2
%         -
%         2\langle x_i,\bar x\rangle
%         +
%         \|\bar x\|^2
%     \right)  =
%     \sum_i\|x_i\|^2
%     -
%     2\left\langle \sum_i x_i,\bar x\right\rangle
%     +
%     n\|\bar x\|^2.
% \end{aligned}
% \]
% Using \(\sum_i x_i=n\bar x\), the middle term becomes
% \[
%     2\left\langle \sum_i x_i,\bar x\right\rangle
%     =
%     2n\|\bar x\|^2.
% \]
% Therefore,
% \[
%     \sum_i\|x_i-\bar x\|^2
%     =
%     \sum_i\|x_i\|^2
%     -
%     n\|\bar x\|^2.
% \]
Multiplying both sides of the equation in \eqref{eq:252} by \(2n\), we obtain
\[
    2n\sum_i\|x_i-\bar x\|^2
    =
    2n\sum_i\|x_i\|^2
    -
    2n^2\|\bar x\|^2.
\]
Comparing this with the previous expression, we conclude that
\[
    \sum_{i,j}\|x_i-x_j\|^2
    = 2n\sum_i\|x_i\|^2
    -
    2n^2\|\bar x\|^2=
    2n\sum_i\|x_i-\bar x\|^2.
\]
Consequently,
\[
    \|SX\|_2^2
    =
    \frac{1}{n^2}\sum_{i,j}\|x_i-x_j\|^2
    =
    \frac{2}{n}\sum_i\|x_i-\bar x\|^2
    \le
    \frac{2}{n}\sum_i\|x_i\|^2
    =
    2\|X\|_2^2,
\]
where the inequality follows from \eqref{eq:252}.
%\(\sum_i\|x_i-\bar x\|^2=\sum_i\|x_i\|^2-n\|\bar x\|^2\le \sum_i\|x_i\|^2\). 
Hence
\[
    \|SX\|_2\le \sqrt{2}\|X\|_2.
\]

For \(p=+\infty\), for every pair \(i,j\),
\[
    \|x_i-x_j\|
    \le
    \|x_i\|+\|x_j\|
    \le
    2\|X\|_\infty.
\]
Taking the maximum over \(i,j\) gives
\[
    \|SX\|_\infty\le 2\|X\|_\infty.
\]

It remains to interpolate the endpoint bounds. 
%We use the standard Riesz--Thorin interpolation theorem in this finite-dimensional setting: 
We use the vector-valued version of the Riesz--Thorin interpolation theorem,
applied to the linear operators \(T\) and \(S\)% acting on\(\ell_p^n(\mathbb R^2)\) and \(\ell_p^{n^2}(\mathbb R^2)\)
:
if a linear operator \(A\) satisfies
\(\|AX\|_r\le C_r\|X\|_r\) and \(\|AX\|_s\le C_s\|X\|_s\), then for
\(p\) satisfying
\[
    \frac{1}{p}
    =
    \frac{1-\lambda}{r}+\frac{\lambda}{s}
    \qquad
    \text{for some } \lambda\in[0,1],
\]
we have
\[
    \|AX\|_p\le C_r^{1-\lambda}C_s^\lambda\|X\|_p.
\]

For \(T\) and \(1\le p\le 2\), interpolate between the \(L_1\) bound
\(C_1=2q_n\) and the \(L_2\) bound \(C_2=1\). Here
\[
    \frac{1}{p}
    =
    \frac{1-\lambda}{1}+\frac{\lambda}{2},
    \qquad\text{so}\qquad
    \lambda=2-\frac{2}{p}.
\]
Therefore,
\[
    \|TX\|_p
    \le
    (2q_n)^{1-\lambda}\|X\|_p
    =
    (2q_n)^{\frac{2}{p}-1}\|X\|_p.
\]

For \(S\) and \(1\le p\le 2\), interpolate between \(C_1=2q_n\) and
\(C_2=\sqrt{2}\), with the same
\(\lambda=2-\frac{2}{p}\). This gives
\[
\begin{aligned}
    \|SX\|_p
    &\le
    (2q_n)^{1-\lambda}(\sqrt{2})^\lambda\|X\|_p  =
    (2q_n)^{\frac{2}{p}-1}
    2^{\frac{1}{2}\left(2-\frac{2}{p}\right)}
    \|X\|_p  =
    2^{\frac{1}{p}}q_n^{\frac{2}{p}-1}\|X\|_p.
\end{aligned}
\]

For \(T\) and \(2\le p\le+\infty\), interpolate between the \(L_2\) bound
\(C_2=1\) and the \(L_\infty\) bound \(C_\infty=2q_n\). In this case,
\[
    \frac{1}{p}
    =
    \frac{1-\lambda}{2},
    \qquad\text{so}\qquad
    \lambda=1-\frac{2}{p}.
\]
Thus
\[
    \|TX\|_p
    \le
    (2q_n)^\lambda\|X\|_p
    =
    (2q_n)^{1-\frac{2}{p}}\|X\|_p.
\]

Finally, for \(S\) and \(2\le p\le+\infty\), interpolate between
\(C_2=\sqrt{2}\) and \(C_\infty=2\), again with
\(\lambda=1-\frac{2}{p}\). We get
\[
\begin{aligned}
    \|SX\|_p
    &\le
    (\sqrt{2})^{1-\lambda}2^\lambda\|X\|_p  =
    2^{\frac{1}{2}\left(\frac{2}{p}\right)}
    2^{1-\frac{2}{p}}\|X\|_p  =
    2^{1-\frac{1}{p}}\|X\|_p.
\end{aligned}
\]
This proves all the stated inequalities.
\end{proof}

Now we bound the approximation ratio. The proof idea is as follows. First, translate the instance so that the optimal facility is at the origin,  and represent the profile by the agents’ translated locations. Second, the expected cost of CRD naturally decomposes into two terms: the cost of the centroid part and the cost of the random-dictatorship part. The first term is exactly controlled by the centering operator $T$, and the second term is controlled by the pairwise-difference operator $S$.

\begin{lemma}\label{lem:crd-finite-p}
The approximation ratio of the CRD mechanism is bounded by
\[
    \alpha_p(\mathrm{CRD})
    \le
    \begin{cases}
        \displaystyle
        \frac{1}{2}q_n^{\frac{2}{p}-1}
        \left(
            2^{\frac{2}{p}-1}+2^{\frac{1}{p}}
        \right),
        & 1\le p\le 2,\\[12pt]
        \displaystyle
        \frac{1}{2}
        \left(
            2^{1-\frac{2}{p}}q_n^{1-\frac{2}{p}}
            +
            2^{1-\frac{1}{p}}
        \right),
        & 2\le p\le +\infty.\\[12pt]
    \end{cases}
\]
\end{lemma}

\begin{proof}
When $p=+\infty$, \cite{tang2020characterization} provides an approximation of $2-\frac{1}{n}$, and we focus on finite $p$.
It is convenient to use the normalized $L_p$-norm social cost
\[
    \overline{\mathrm{SC}}_p(\mathbf{x},y)
    =
    \left(
        \frac{1}{n}\sum_i \|x_i-y\|^p
    \right)^{\frac{1}{p}}.
\]
This normalization does not change approximation ratios, because
\(
    \overline{\mathrm{SC}}_p(\mathbf{x},y)
    =
    n^{-\frac{1}{p}}\mathrm{SC}_p(\mathbf{x},y)
\)
for every \(y\), and
\(\overline{\mathrm{OPT}}_p(\mathbf{x})=
n^{-\frac{1}{p}}\mathrm{OPT}_p(\mathbf{x})\).

Let \(o\in\mathbb R^2\) be an optimal facility location. We translate the profile $\mathbf x$ by defining the translated vector \(x_i'=x_i-o\) for each agent $i$, so that the optimal facility is at the origin. Let \(X=(x_1',\ldots,x_n')\) represent the profile of translated locations of agents.
 Then
\[
    \|X\|_p
    =
    \left(
        \frac{1}{n}\sum_i\|x_i'\|^p
    \right)^{\frac{1}{p}}
    =
    \overline{\mathrm{OPT}}_p(\mathbf{x}).
\]
Let \(\bar x'=\frac{1}{n}\sum_i x_i'\). Since the centroid of the original
profile is
\[
    \bar x
    =
    \frac{1}{n}\sum_i x_i
    =
    \frac{1}{n}\sum_i(o+x_i')
    =
    o+\bar x',
\]
the distance from agent \(i\) to the centroid is
\[
    \|x_i-\bar x\|
    =
    \|o+x_i'-(o+\bar x')\|
    =
    \|x_i'-\bar x'\|.
\]
Therefore the normalized cost of the centroid part is exactly
\[
\begin{aligned}
    \overline{\mathrm{SC}}_p(\mathbf{x},\bar x)
    &=
    \left(
        \frac{1}{n}\sum_i \|x_i'-\bar x'\|^p
    \right)^{\frac{1}{p}}  =
    \|TX\|_p,
\end{aligned}
\]
where \((TX)_i=x_i'-\bar x'\).

Now consider the random-dictatorship part. If the mechanism returns agent
\(j\)'s location \(x_j\), then the normalized cost is
\[
\begin{aligned}
    \overline{\mathrm{SC}}_p(\mathbf{x},x_j)
    &=
    \left(
        \frac{1}{n}\sum_i \|x_i-x_j\|^p
    \right)^{\frac{1}{p}}  =
    \left(
        \frac{1}{n}\sum_i \|x_i'-x_j'\|^p
    \right)^{\frac{1}{p}}.
\end{aligned}
\]
Thus, conditional on using the random-dictatorship part, the expected
normalized cost is
\[
    D
    :=
    \frac{1}{n}\sum_j
    \left(
        \frac{1}{n}\sum_i \|x_i'-x_j'\|^p
    \right)^{\frac{1}{p}}.
\]
We next upper bound \(D\) by \(\|SX\|_p\). For each \(j\), define
\[
    A_j
    :=
    \left(
        \frac{1}{n}\sum_i \|x_i'-x_j'\|^p
    \right)^{\frac{1}{p}}.
\]
Then \(D=\frac{1}{n}\sum_j A_j\). Since the function \(r\mapsto r^p\) is
convex for \(p\ge 1\), Jensen's inequality gives
\[
    \left(\frac{1}{n}\sum_j A_j\right)^p
    \le
    \frac{1}{n}\sum_j A_j^p.
\]
Taking the \(p\)-th root on both sides,
\[
\begin{aligned}
    D
    &=
    \frac{1}{n}\sum_j A_j
    \le
    \left(
        \frac{1}{n}\sum_j A_j^p
    \right)^{\frac{1}{p}} =
    \left(
        \frac{1}{n}\sum_j
        \frac{1}{n}\sum_i \|x_i'-x_j'\|^p
    \right)^{\frac{1}{p}}  \\
    &=
    \left(
        \frac{1}{n^2}\sum_{i,j}\|x_i'-x_j'\|^p
    \right)^{\frac{1}{p}}
    =
    \|SX\|_p,
\end{aligned}
\]
where \((SX)_{i,j}=x_i'-x_j'\).

The mechanism \(\mathrm{CRD}\) chooses the centroid $\bar x$ with probability
\(\frac{1}{2}\), and chooses the random-dictatorship part with probability
\(\frac{1}{2}\). Hence its expected normalized cost satisfies
\[
\begin{aligned}
    \mathbb E\left[
        \overline{\mathrm{SC}}_p(\mathbf{x},\mathrm{CRD}(\mathbf{x}))
    \right]
    &=
    \frac{1}{2}\overline{\mathrm{SC}}_p(\mathbf{x},\bar x)
    +
    \frac{1}{2}D  \le
    \frac{1}{2}\|TX\|_p+\frac{1}{2}\|SX\|_p.
\end{aligned}
\]

We now apply Lemma~\ref{lem:crd-operator}. If \(1\le p\le 2\), then
\[
    \|TX\|_p
    \le
    (2q_n)^{\frac{2}{p}-1}\|X\|_p
    =
    2^{\frac{2}{p}-1}q_n^{\frac{2}{p}-1}\|X\|_p,
\]
and
\[
    \|SX\|_p
    \le
    2^{\frac{1}{p}}q_n^{\frac{2}{p}-1}\|X\|_p.
\]
Therefore,
\[
\begin{aligned}
    \mathbb E\left[
        \overline{\mathrm{SC}}_p(\mathbf{x},\mathrm{CRD}(\mathbf{x}))
    \right]
    &\le
    \frac{1}{2}
    \left(
        2^{\frac{2}{p}-1}q_n^{\frac{2}{p}-1}
        +
        2^{\frac{1}{p}}q_n^{\frac{2}{p}-1}
    \right)\|X\|_p  \\
    &=
    \frac{1}{2}q_n^{\frac{2}{p}-1}
    \left(
        2^{\frac{2}{p}-1}+2^{\frac{1}{p}}
    \right)
    \overline{\mathrm{OPT}}_p(\mathbf{x}).
\end{aligned}
\]
This gives the claimed bound for \(1\le p\le 2\).

If \(2\le p<+\infty\), Lemma~\ref{lem:crd-operator} gives
\[
    \|TX\|_p
    \le
    (2q_n)^{1-\frac{2}{p}}\|X\|_p
    =
    2^{1-\frac{2}{p}}q_n^{1-\frac{2}{p}}\|X\|_p,
\]
and
\[
    \|SX\|_p
    \le
    2^{1-\frac{1}{p}}\|X\|_p.
\]
Thus,
\[
\begin{aligned}
    \mathbb E\left[
        \overline{\mathrm{SC}}_p(\mathbf{x},\mathrm{CRD}(\mathbf{x}))
    \right]
    &\le
    \frac{1}{2}
    \left(
        2^{1-\frac{2}{p}}q_n^{1-\frac{2}{p}}
        +
        2^{1-\frac{1}{p}}
    \right)\|X\|_p  \\
    &=
    \frac{1}{2}
    \left(
        2^{1-\frac{2}{p}}q_n^{1-\frac{2}{p}}
        +
        2^{1-\frac{1}{p}}
    \right)
    \overline{\mathrm{OPT}}_p(\mathbf{x}).
\end{aligned}
\]
This proves the claimed bound for \(2\le p<+\infty\).
%Finally, since normalized and unnormalized costs differ by the same factor \(n^{-\frac{1}{p}}\) for both the algorithmic cost and the optimum, the same approximation ratios hold for the original $L_p$-norm social cost \(\mathrm{SC}_p\).
\end{proof}

% \begin{lemma}\label{lem:crd-infinity}\cite{tang2020characterization}
% For \(p=+\infty\), \(\mathrm{CRD}\) is a tight
% \(\left(2-\frac{1}{n}\right)\)-approximation.\redcomment{[write together?] \bluecomment{ok}}
% \end{lemma}

% \begin{proof}
% Let \(o\) be an optimal facility for the maximum cost, and write
% \(x_i'=x_i-o\). Then \(\max_i\|x_i'\|=\mathrm{OPT}_{+\infty}(\mathbf{x})\).
% By Lemma~\ref{lem:crd-operator},
% \[
%     \|TX\|_\infty\le 2q_n\|X\|_\infty,
%     \qquad
%     \|SX\|_\infty\le 2\|X\|_\infty.
% \]
% The centroid part has maximum cost \(\|TX\|_\infty\), and the
% random-dictatorship part has expected maximum cost at most \(\|SX\|_\infty\).
% Hence
% \[
%     \mathrm{ALG}_{+\infty}(\mathbf{x})
%     \le
%     \frac{1}{2}\cdot 2q_n\,\mathrm{OPT}_{+\infty}(\mathbf{x})
%     +
%     \frac{1}{2}\cdot 2\,\mathrm{OPT}_{+\infty}(\mathbf{x})
%     =
%     \left(2-\frac{1}{n}\right)\mathrm{OPT}_{+\infty}(\mathbf{x}).
% \]

% For tightness, take \(n-1\) agents at \(0\) and one agent at \(1\) on a
% line. The optimal maximum cost is \(\frac{1}{2}\). The centroid is
% \(\frac{1}{n}\), whose maximum cost is \(1-\frac{1}{n}=q_n\). Returning any
% agent location gives maximum cost \(1\). Therefore
% \[
%     \mathrm{ALG}_{+\infty}(\mathbf{x})
%     =
%     \frac{1}{2}q_n+\frac{1}{2},
% \]
% and the ratio is
% \[
%     \frac{\frac{1}{2}q_n+\frac{1}{2}}{\frac{1}{2}}
%     =
%     1+q_n
%     =
%     2-\frac{1}{n}.
% \]
% \end{proof}

We now record the tight examples for endpoint cases.

\begin{lemma}\label{lem:crd-tight-endpoints}
The bounds for \(\mathrm{CRD}\) in
Lemma~\ref{lem:crd-finite-p} are tight for every \(n\) at \(p=1,+\infty\) and for every even $n$ at
\(p=2\). Moreover, for \(n=2\), the bounds are tight for all $p$. 
\end{lemma}

\begin{proof}
For \(p=1\), take \(n-1\) agents at \(0\) and one agent at \(1\) on a line.
The optimal $L_1$ social cost is \(1\). The centroid is \(\frac{1}{n}\), whose
social cost is \(2-\frac{2}{n}\). The random-dictatorship part has expected
social cost
\[
    \frac{n-1}{n}\cdot 1+\frac{1}{n}\cdot (n-1)
    =
    2-\frac{2}{n}.
\]
Thus \(\mathrm{CRD}\) has ratio \(2-\frac{2}{n}\), matching the bound at
\(p=1\).

% For \(p=2\), place the agents at the vertices of a regular \(n\)-gon on the
% unit circle centered at the origin. The centroid is the origin, which is also
% the \(L_2\)-optimal facility. The optimal \(L_2\) social cost is \(\sqrt n\).
% If the mechanism returns any vertex, the \(L_2\) social cost is
% \(\sqrt{2n}\). Hence the ratio of \(\mathrm{CRD}\) is
% \[
%     \frac{1}{2}\cdot 1+\frac{1}{2}\cdot \sqrt{2}
%     =
%     \frac{1+\sqrt{2}}{2},
% \]
% matching the bound at \(p=2\).
For $p=2$ and any even $n$, consider the instance with $\frac n2$ agents at 0 and $\frac n2$ agents at 1. The centroid $\frac12$ is optimal and
\(\mathrm{OPT}_p=(\frac{n}{2^p})^{\frac1p}=\frac12n^{\frac1p}\); however, returning any agent location has social cost of \((\frac n2)^{\frac1p}\). Therefore
the ratio is
\[
    \frac{
        \frac{1}{2}\mathrm{OPT}_p+\frac{1}{2}(\frac n2)^{\frac1p}
    }{
        \mathrm{OPT}_p
    }
    =
    \frac{1+\sqrt 2}{2},
\]
matching the upper bound in Lemma~\ref{lem:crd-finite-p}.

For $p=+\infty$, we still take $n-1$ agents at 0 and one agent at 1 on a line with the centroid at $\frac{1}{n}$. The optimal maximum cost is $\frac{1}{2}$ and the mechanism's expected maximum cost is
\[
    \frac{1}{2}\cdot 1 + \frac{1}{2}\cdot \left(1-\frac{1}{n}\right) = \frac{1}{2}\left(2-\frac{1}{n}\right),
\]
achieving an approximation ratio of $2-\frac{1}{n}$.

Finally, we consider \(n=2\), with the two agents at distance \(1\). For every
finite \(p\), the midpoint/centroid is optimal and
\(\mathrm{OPT}_p=2^{\frac{1}{p}-1}\); however, returning either agent has social cost of \(1\). Therefore
the ratio is
\[
    \frac{
        \frac{1}{2}\mathrm{OPT}_p+\frac{1}{2}
    }{
        \mathrm{OPT}_p
    }
    =
    \frac{1+2^{1-\frac{1}{p}}}{2}.
\]
When \(n=2\), the upper bound in Lemma~\ref{lem:crd-finite-p} is exactly
\(\frac{1+2^{1-\frac{1}{p}}}{2}\), so it is tight.
\end{proof}

\begin{proof}[Proof of Theorem~\ref{thm:crd}]
Strategyproofness in expectation was proved by Tang et al.
\cite{tang2020characterization}. Intuitively, when an agent moves its report,
the possible decrease in its expected distance to the centroid is exactly
balanced by the additional probability that the mechanism returns the
agent's reported location.

The approximation guarantees follow from
Lemma~\ref{lem:crd-finite-p}. Tightness
follows from
Lemma~\ref{lem:crd-tight-endpoints}.
\end{proof}

\section{Symmetric Monotone Norm Objectives}\label{sec:general-norm}

In this section, we go beyond \(L_p\)-norm social costs and consider
objectives induced by symmetric monotone norms on the distance vector. Let
\(g:\mathbb R^n\to\mathbb R_+\) be a norm. We say that \(g\) is
\emph{monotone} if \(g(\mathbf v)\le g(\mathbf u)\) whenever
\(0\le \mathbf v\le \mathbf u\) coordinatewise. We say that \(g\) is
\emph{symmetric} if \(g(\mathbf v)=g(\mathbf u)\) whenever \(\mathbf u\) is
obtained from \(\mathbf v\) by permuting coordinates and changing signs.
Equivalently, \(g\) depends only on the multiset of absolute values of the
coordinates. Since the vectors considered below are distance vectors, they
are always nonnegative.

For a profile \(\mathbf{x}\) and a facility location \(y\), write
\[
    D(\mathbf{x},y)
    =
    \bigl(\|x_1-y\|,\ldots,\|x_n-y\|\bigr)
\]
for the vector of distances from \(y\) to the agents. The corresponding
symmetric monotone norm objective is
\[
    \SC_g(\mathbf{x},y)
    =
    g(D(\mathbf{x},y)),
    \qquad
    \OPT_g(\mathbf{x})
    =
    \min_{z} g(D(\mathbf{x},z)).
\]
This class contains the \(L_p\)-norm social costs as special cases, but also
includes many other objectives that depend on the ordered distance profile.

For \(\ell\in[n]\), let \(\Top_\ell:\mathbb R^n\to\mathbb R_+\) denote the
sum of the \(\ell\) largest absolute values of the coordinates. In particular,
when the input vector is nonnegative, \(\Top_\ell(\mathbf v)\) is simply the
sum of the \(\ell\) largest entries of \(\mathbf v\). The cases \(\ell=1\)
and \(\ell=n\) recover the \(L_\infty\) and \(L_1\) norms, respectively.
We use the following standard consequence of the Hardy--Littlewood--Pólya
majorization theorem.

\begin{lemma}[\cite{1934Inequalities}]\label{lem:hardy}
Let \(\mathbf v,\mathbf u\in\mathbb R_+^n\) and \(\alpha\ge0\). If
\[
    \Top_\ell(\mathbf v)\le \alpha\,\Top_\ell(\mathbf u)
    \qquad\text{for every } \ell\in[n],
\]
then \(g(\mathbf v)\le \alpha\,g(\mathbf u)\) for every symmetric monotone
norm \(g:\mathbb R^n\to\mathbb R_+\).
\end{lemma}

\begin{lemma}\label{lem:cm-topell-domination}
Let \(\mathbf{x}\in(\mathbb R^2)^n\), let \(m=\CM(\mathbf{x})\), and let
distances be measured by the Euclidean norm. Then, for every comparison point
\(z\in\mathbb R^2\) and every \(\ell\in[n]\),
\[
    \Top_{\ell}\!\left(D(\mathbf{x},m)\right)
    \le
    2\,\Top_{\ell}\!\left(D(\mathbf{x},z)\right).
\]
\end{lemma}

\begin{proof}
By translating the instance, assume without loss of generality that
\(m=\CM(\mathbf{x})=(0,0)\).

We first reduce to the case where \(n\) is even. If \(n\) is odd, duplicate
every agent and denote the resulting profile by \(\mathbf{x}^{(2)}\). The
coordinate-wise median remains \(m\). Moreover, for every facility location
\(y\) and every \(\ell\in[n]\),
\[
    \Top_{2\ell}\!\left(D(\mathbf{x}^{(2)},y)\right)
    =
    2\,\Top_{\ell}\!\left(D(\mathbf{x},y)\right).
\]
Thus the desired inequality for \((\mathbf{x},\ell)\) follows from the even
case applied to \((\mathbf{x}^{(2)},2\ell)\). Hence, in the rest of the proof,
assume that \(n\) is even.

Since \(m=(0,0)\) is the coordinate-wise median, for each coordinate at most
\(n/2\) agents lie strictly on either side of the origin. Assign artificial
signs \(+\) and \(-\) to zero coordinates so that, in each coordinate, exactly
\(n/2\) agents have sign \(+\) and exactly \(n/2\) agents have sign \(-\).
This partitions the agents into four signed classes
\[
    (++),\qquad (+-),\qquad (-+),\qquad (--).
\]
The number of \((++)\) agents equals the number of \((--)\) agents, and the
number of \((+-)\) agents equals the number of \((-+)\) agents. Therefore, we
can pair every \((++)\) agent with a \((--)\) agent, and every \((+-)\) agent
with a \((-+)\) agent.

For every resulting pair of locations \((a,b)\in(\mathbb R^2)^2\), we have
\(a_1b_1\le0\) and \(a_2b_2\le0\), and hence \(a\cdot b\le0\). Therefore,
\[
    \|a-b\|^2
    =
    \|a\|^2+\|b\|^2-2a\cdot b
    \ge
    \|a\|^2+\|b\|^2.
\]
In particular,
\[
    \max\{\|a\|,\|b\|\}\le \|a-b\|.
\]
By the triangle inequality, for every comparison point \(z\in\mathbb R^2\),
\[
    \|a-b\|
    =
    \|(a-z)-(b-z)\|
    \le
    \|a-z\|+\|b-z\|.
\]

Now fix \(\ell\in[n]\), and let \(A\subseteq N\) be a set of \(\ell\) agents
with the largest distances from \(m\). Since \(m=0\),
\[
    \Top_{\ell}\!\left(D(\mathbf{x},m)\right)
    =
    \sum_{i\in A}\|x_i\|.
\]
We construct a set \(B\subseteq N\) of \(\ell\) agents whose distances to
\(z\) will be used for charging. Consider each paired pair \(\{a,b\}\). If no
agent in the pair belongs to \(A\), choose no agent from this pair. If exactly
one agent in the pair belongs to \(A\), choose from this pair an agent with
larger distance to \(z\). If both agents in the pair belong to \(A\), choose
both agents. By construction, \(|B|=|A|=\ell\).

We claim that, for each pair, the contribution of agents in \(A\) from this
pair is at most twice the contribution of agents in \(B\) from this pair,
measured with respect to \(z\). If exactly one of \(a,b\) belongs to \(A\),
say \(a\), then
\[
    \|a\|
    \le
    \max\{\|a\|,\|b\|\}
    \le
    \|a-b\|
    \le
    \|a-z\|+\|b-z\|
    \le
    2\max\{\|a-z\|,\|b-z\|\}.
\]
The last term is twice the distance to \(z\) of the agent chosen in \(B\)
from this pair. If both \(a,b\) belong to \(A\), then both are chosen in
\(B\), and
\[
    \|a\|+\|b\|
    \le
    2\|a-b\|
    \le
    2\bigl(\|a-z\|+\|b-z\|\bigr).
\]
Summing over all pairs, we obtain
\[
    \Top_{\ell}\!\left(D(\mathbf{x},m)\right)
    =
    \sum_{i\in A}\|x_i\|
    \le
    2\sum_{i\in B}\|x_i-z\|.
\]
Since \(|B|=\ell\), the sum over \(B\) is at most the sum of the \(\ell\)
largest entries of \(D(\mathbf{x},z)\). Hence
\[
    \Top_{\ell}\!\left(D(\mathbf{x},m)\right)
    \le
    2\,\Top_{\ell}\!\left(D(\mathbf{x},z)\right).
\]
\end{proof}

The extension from Top-\(\ell\) objectives to arbitrary symmetric monotone
norm objectives follows from Lemma~\ref{lem:hardy}.

\begin{theorem}\label{thm:cm-symmetric-monotone}
In the Euclidean plane, the coordinate-wise median mechanism is a
\(2\)-approximation under every symmetric monotone norm objective. That is,
for every symmetric monotone norm \(g:\mathbb R^n\to\mathbb R_+\) and every
profile \(\mathbf{x}\in(\mathbb R^2)^n\),
\[
    g\!\left(D(\mathbf{x},\CM(\mathbf{x}))\right)
    \le
    2\min_{z\in\mathbb R^2} g\!\left(D(\mathbf{x},z)\right).
\]
Moreover, the factor \(2\) is tight as a uniform guarantee over all symmetric
monotone norm objectives.
\end{theorem}

\begin{proof}
Let \(m=\CM(\mathbf{x})\), and let \(z\in\mathbb R^2\) be arbitrary. By
Lemma~\ref{lem:cm-topell-domination}, for every \(\ell\in[n]\),
\[
    \Top_{\ell}\!\left(D(\mathbf{x},m)\right)
    \le
    2\,\Top_{\ell}\!\left(D(\mathbf{x},z)\right).
\]
Applying Lemma~\ref{lem:hardy}, we obtain
\[
    g\!\left(D(\mathbf{x},m)\right)
    \le
    2g\!\left(D(\mathbf{x},z)\right).
\]
Since this holds for every comparison point \(z\in\mathbb R^2\), taking \(z\)
to be an optimal facility location for the \(g\)-objective proves the upper
bound.

For tightness, consider the \(L_\infty\)-norm objective on the distance
vector, which is a symmetric monotone norm objective. Take a profile with
\(n-1\) agents at the origin and one agent at a nonzero point \(u\). For
\(n\ge3\), CM outputs the origin. The mechanism's cost is \(\|u\|\), whereas
placing the facility at \(u/2\) has maximum distance \(\|u\|/2\). Hence the
ratio is \(2\), so the uniform factor \(2\) cannot be improved.
\end{proof}

We next consider the uniformly rotated coordinate-wise median mechanism.

\begin{theorem}\label{thm:urcm-symmetric-monotone}
In the Euclidean plane, the uniformly rotated coordinate-wise median mechanism
is a \(2\)-approximation under every symmetric monotone norm objective. That
is, for every symmetric monotone norm \(g:\mathbb R^n\to\mathbb R_+\) and every
profile \(\mathbf{x}\in(\mathbb R^2)^n\),
\[
    \mathbb E\left[
        g\left(D(\mathbf{x},\mathrm{URCM}(\mathbf{x}))\right)
    \right]
    \le
    2\min_{z\in\mathbb R^2}g\left(D(\mathbf{x},z)\right).
\]
Moreover, the factor \(2\) is tight as a uniform guarantee over all symmetric
monotone norm objectives.
\end{theorem}

\begin{proof}
For an angle \(\theta\in[0,2\pi)\), let \(R_\theta:\mathbb R^2\to\mathbb R^2\)
be the orthogonal map that writes each point in the rotated basis
\(\{e_\theta,e_\theta^\perp\}\), i.e.,
\[
    R_\theta x
    =
    \left(\langle x,e_\theta\rangle,\langle x,e_\theta^\perp\rangle\right).
\]
For the realization \(\Theta=\theta\), the output of \(\mathrm{URCM}\) is
\(R_\theta^{-1}\mathrm{CM}(R_\theta\mathbf{x})\), where
\(R_\theta\mathbf{x}=(R_\theta x_1,\ldots,R_\theta x_n)\).

Fix an arbitrary comparison point \(z\in\mathbb R^2\). Since \(R_\theta\) is an
isometry for the Euclidean norm, we have
\[
    D\!\left(\mathbf{x},
    R_\theta^{-1}\mathrm{CM}(R_\theta\mathbf{x})\right)
    =
    D\!\left(R_\theta\mathbf{x},\mathrm{CM}(R_\theta\mathbf{x})\right)
\]
and
\[
    D(\mathbf{x},z)
    =
    D(R_\theta\mathbf{x},R_\theta z).
\]
Therefore, by Theorem~\ref{thm:cm-symmetric-monotone}, for every fixed
\(\theta\),
\[
    g\!\left(
        D\!\left(\mathbf{x},
        R_\theta^{-1}\mathrm{CM}(R_\theta\mathbf{x})\right)
    \right)
    \le
    2g\!\left(D(\mathbf{x},z)\right).
\]
Taking expectation over \(\Theta\) and then taking \(z\) to be an optimal
facility location for the \(g\)-objective gives the desired upper bound.

The tightness follows from the same \(L_\infty\)-objective instance as in the
proof of Theorem~\ref{thm:cm-symmetric-monotone}. Let \(u\in\mathbb R^2\) be
nonzero, and consider the profile with \(n-1\) agents at the origin and one
agent at \(u\). For every rotation angle \(\theta\), the rotated profile has
\(n-1\) agents at the origin and one agent at \(R_\theta u\), so for
\(n\ge3\) the coordinate-wise median in the rotated coordinate system is the
origin. Hence \(\mathrm{URCM}\) always outputs the origin. Under the
\(L_\infty\)-norm objective on the distance vector, the mechanism has cost
\(\|u\|\), whereas the optimal facility location \(u/2\) has cost
\(\|u\|/2\). Thus the ratio is \(2\).
\end{proof}

For the next result, we use a more general location space. Let \(V\) be a
real vector space. A function \(\|\cdot\|_V:V\to\mathbb R_+\) is a norm if
for all \(u,v\in V\) and all \(\lambda\in\mathbb R\), it satisfies
\[
    \|v\|_V=0 \text{ if and only if } v=0,\qquad
    \|\lambda v\|_V=|\lambda|\,\|v\|_V,\qquad
    \|u+v\|_V\le \|u\|_V+\|v\|_V .
\]
We call \((V,\|\cdot\|_V)\) a \emph{real normed vector space}. It induces the
metric \(d_V(x,y)=\|x-y\|_V\). No finite-dimensionality or inner-product
structure is assumed. For a profile \(\mathbf{x}\in V^n\) and a facility
location \(y\in V\), write
\[
    D_V(\mathbf{x},y)
    =
    \bigl(\|x_1-y\|_V,\ldots,\|x_n-y\|_V\bigr).
\]
For a symmetric monotone norm \(g:\mathbb R^n\to\mathbb R_+\), define
\[
    \SC_{g,V}(\mathbf{x},y)
    =
    g(D_V(\mathbf{x},y)),
    \qquad
    \OPT_{g,V}(\mathbf{x})
    =
    \inf_{z\in V} g(D_V(\mathbf{x},z)).
\]
We use an infimum rather than a minimum because an optimal facility location
need not exist in an arbitrary normed vector space.

The centroid random dictatorship mechanism is still well-defined in this
setting: given \(\mathbf{x}\in V^n\), it returns the arithmetic centroid
\(\bar x=\frac1n\sum_i x_i\) with probability \(\frac12\), and returns each
agent location \(x_j\) with probability \(\frac{1}{2n}\).

\begin{lemma}\label{lem:average-vector-domination}
Let \(\mathbf d=(d_1,\ldots,d_n)\in\mathbb R_+^n\), and let
\(S=\sum_i d_i\). Then, for every symmetric monotone norm
\(g:\mathbb R^n\to\mathbb R_+\),
\[
    \frac{S}{n}g(\mathbf 1)\le g(\mathbf d),
\]
where \(\mathbf 1=(1,\ldots,1)\).
\end{lemma}

\begin{proof}
For every \(\ell\in[n]\), the sum of the largest \(\ell\) entries of
\(\mathbf d\) is at least \(\ell\) times the average entry. Hence
\[
    \Top_\ell\left(\frac{S}{n}\mathbf 1\right)
    =
    \frac{\ell S}{n}
    \le
    \Top_\ell(\mathbf d).
\]
By Lemma~\ref{lem:hardy}, we have
\(g(\frac{S}{n}\mathbf 1)\le g(\mathbf d)\), which is exactly the desired
inequality by homogeneity of \(g\).
\end{proof}

\begin{theorem}\label{thm:crd-normed-symmetric-monotone}
Let \((V,\|\cdot\|_V)\) be any real normed vector space. The centroid random
dictatorship mechanism is strategyproof in \(V\). Moreover, for every
symmetric monotone norm \(g:\mathbb R^n\to\mathbb R_+\) and every profile
\(\mathbf{x}\in V^n\),
\[
    \mathbb E\left[
        g\left(D_V(\mathbf{x},\mathrm{CRD}(\mathbf{x}))\right)
    \right]
    \le
    \left(2-\frac1n\right)\OPT_{g,V}(\mathbf{x}).
\]
If \(V\) contains a nonzero vector, then the factor \(2-\frac1n\) is tight for
every \(n\ge2\).
\end{theorem}

\begin{proof}
The case \(n=1\) is trivial, so assume \(n\ge2\).

We first prove strategyproofness. Fix an agent \(i\), its true location
\(x_i\in V\), and the reports of all other agents. Let
\(\bar x=\frac1n\sum_j x_j\) be the truthful centroid. Suppose agent \(i\)
misreports \(x_i'=x_i+\delta\), where \(\delta\in V\). The new centroid is
\(\bar x'=\bar x+\frac1n\delta\). The only parts of the expected cost of
agent \(i\) affected by this misreport are the centroid outcome and the event
that agent \(i\)'s own reported location is selected as the dictator. Hence it
suffices to show
\[
    \frac12\|x_i-\bar x\|_V
    \le
    \frac12\|x_i-\bar x'\|_V
    +
    \frac{1}{2n}\|x_i-x_i'\|_V .
\]
Since \(x_i-\bar x=(x_i-\bar x')+\frac1n\delta\), the triangle inequality and
homogeneity of the norm give
\[
    \|x_i-\bar x\|_V
    \le
    \|x_i-\bar x'\|_V+\frac1n\|\delta\|_V
    =
    \|x_i-\bar x'\|_V+\frac1n\|x_i-x_i'\|_V.
\]
Multiplying by \(\frac12\) gives the desired inequality. Thus CRD is
strategyproof.

We now prove the approximation guarantee. Fix an arbitrary comparison point
\(z\in V\). Let
\[
    d_i=\|x_i-z\|_V,\qquad
    \mathbf d=(d_1,\ldots,d_n),\qquad
    S=\sum_i d_i.
\]
We separately bound the centroid part and the random-dictatorship part.

Let \(\bar x=\frac1n\sum_j x_j\) be the centroid. For every agent \(i\),
\[
    x_i-\bar x
    =
    x_i-\frac1n\sum_{j=1}^n x_j
    =
    \frac1n\sum_{j\ne i}(x_i-x_j).
\]
Hence, by the triangle inequality and homogeneity in \(V\),
\[
    \|x_i-\bar x\|_V
    \le
    \frac1n\sum_{j\ne i}\|x_i-x_j\|_V.
\]
For each \(j\ne i\), the triangle inequality with the comparison point \(z\)
gives
\[
    \|x_i-x_j\|_V
    \le
    \|x_i-z\|_V+\|x_j-z\|_V
    =
    d_i+d_j.
\]
Therefore,
\[
    \|x_i-\bar x\|_V
    \le
    \frac1n\sum_{j\ne i}(d_i+d_j)
    =
    \frac1n\left((n-1)d_i+\sum_{j\ne i}d_j\right).
\]
Since \(S=\sum_j d_j\), we have \(\sum_{j\ne i}d_j=S-d_i\). Thus
\[
    \|x_i-\bar x\|_V
    \le
    \frac1n\left((n-1)d_i+S-d_i\right)
    =
    \frac{n-2}{n}d_i+\frac{S}{n}.
\]
Thus, coordinatewise,
\[
    D_V(\mathbf{x},\bar x)
    \le
    \frac{n-2}{n}\mathbf d+\frac{S}{n}\mathbf 1.
\]
By monotonicity, the triangle inequality for \(g\), homogeneity of \(g\), and
Lemma~\ref{lem:average-vector-domination},
\[
    g(D_V(\mathbf{x},\bar x))
    \le
    \frac{n-2}{n}g(\mathbf d)+\frac{S}{n}g(\mathbf 1)
    \le
    \left(2-\frac2n\right)g(\mathbf d).
\]

Next consider the random-dictatorship part. If the dictator is agent \(j\),
then for every \(i\),
\[
    \|x_i-x_j\|_V
    \le
    \|x_i-z\|_V+\|x_j-z\|_V
    =
    d_i+d_j.
\]
Hence
\[
    D_V(\mathbf{x},x_j)\le \mathbf d+d_j\mathbf 1
\]
coordinatewise. Therefore,
\[
    g(D_V(\mathbf{x},x_j))
    \le
    g(\mathbf d)+d_jg(\mathbf 1).
\]
Averaging over the uniformly chosen dictator \(j\) and applying
Lemma~\ref{lem:average-vector-domination} again gives
\[
    \frac1n\sum_{j=1}^n g(D_V(\mathbf{x},x_j))
    \le
    g(\mathbf d)+\frac{S}{n}g(\mathbf 1)
    \le
    2g(\mathbf d).
\]

Since \(\mathrm{CRD}\) chooses the centroid with probability \(\frac12\) and
chooses each dictator with total probability \(\frac12\), we obtain
\[
    \mathbb E\left[
        g(D_V(\mathbf{x},\mathrm{CRD}(\mathbf{x})))
    \right]
    \le
    \frac12\left(2-\frac2n\right)g(\mathbf d)
    +
    \frac12\cdot 2g(\mathbf d)
    =
    \left(2-\frac1n\right)g(\mathbf d).
\]
Since \(z\in V\) was arbitrary, taking the infimum over \(z\) proves the
upper bound.

It remains to prove tightness. Assume that \(V\) contains a nonzero vector
\(u\). Consider the profile with one agent at \(0\in V\) and \(n-1\) agents at
\(u\). Let \(g\) be the \(L_\infty\)-norm on the distance vector. For any
facility location \(y\in V\), the triangle inequality gives
\[
    \|u\|_V
    =
    \|u-y+y\|_V
    \le
    \|u-y\|_V+\|y\|_V
    \le
    2\max\{\|y\|_V,\|u-y\|_V\}.
\]
Thus the optimal \(L_\infty\)-cost is at least \(\frac12\|u\|_V\). This value
is achieved at \(y=\frac12 u\), so \(\OPT_{g,V}(\mathbf{x})=\frac12\|u\|_V\).

The centroid is \(\bar x=\frac{n-1}{n}u\), whose maximum distance to the
agents is \(\frac{n-1}{n}\|u\|_V\). Every dictator location has maximum
distance \(\|u\|_V\). Hence the expected maximum cost of \(\mathrm{CRD}\) is
\[
    \frac12\cdot\frac{n-1}{n}\|u\|_V
    +
    \frac12\cdot \|u\|_V
    =
    \left(1-\frac{1}{2n}\right)\|u\|_V.
\]
Dividing by \(\frac12\|u\|_V\) gives \(2-\frac1n\). Since \(L_\infty\) is a
symmetric monotone norm, the factor is tight.
\end{proof}

\section{Conclusion}\label{sec:conclusion}

We study Euclidean facility location problems under the $L_p$-norm social cost from the perspective of approximate mechanism design without money. 
From this perspective, we investigate the optimal approximation ratios achievable by strategyproof deterministic and randomized mechanisms for general $p$. 
We begin our study with the deterministic coordinate-wise median mechanism, analyzed by Goel and Hann-Caruthers \cite{GoelH23}, who showed that its approximation ratio lies between \(2^{1-\frac{1}{p}}\) and \(2^{\frac{3}{2}-\frac{2}{p}}\) and conjectured that its approximation ratio is \(2^{1-\frac{1}{p}}\) for \(p\ge 2\). 
We address the conjecture and strengthen their results by providing a tight characterization of its approximation ratio for all 
\(p\ge 1\): it is \(\sqrt{2}\) for \(1\le p\le 2\), equals
\(2^{1-\frac{1}{p}}\) for \(2\le p<+\infty\), and becomes \(2\) for
\(p=+\infty\). 
%studied strategyproof facility location in the Euclidean plane under the $L_p$-norm social cost. 
%For the deterministic coordinate-wise median mechanism,
%we gave a tight characterization of its approximation ratio for all
%\(p\ge 1\): it is \(\sqrt{2}\) for \(1\le p\le 2\), equals
%\(2^{1-\frac{1}{p}}\) for \(2\le p<+\infty\), and becomes \(2\) for
%\(p=+\infty\). 
%This closes the gap left by Goel and Hann-Caruthers
%\cite{GoelH23} for \(p\ge 2\). 

We then analyze two randomized strategyproof mechanisms that achieve better approximation ratios than the CM mechanism for various values of $p$. 
We first consider the uniformly rotated coordinate-wise median mechanism, an idea proposed in \cite{GoelH23} and studied independently by \cite{barak2026facility} for $p=1$. 
We show that this mechanism improves upon the deterministic coordinate-wise median mechanism for \(1\le p<2\), yielding in particular a \(\frac{4}{\pi}\)-approximation for the standard total cost as shown in \cite{barak2026facility}. 
We then consider the centroid random dictatorship mechanism, examined by Feldman and Wilf \cite{feldman2013strategyproof} and Tang et al. \cite{tang2020characterization} for locating a facility in the line metric under the $L_2$ social cost and multi-dimensional normed vector spaces under the maximum cost, respectively. 
We show that this mechanism obtains a better approximation ratio for every finite \(p\gtrsim 1.6\) than the coordinate-wise median mechanisms with or without rotations. 
%of Tang
%et al. \cite{tang2020characterization}, originally studied for maximum cost,
%was extended to all \(p\ge 1\), and gives stronger guarantees near \(p=2\)
%and throughout the range \(p>2\). 
Together, these results demonstrate the extent to which the performance of deterministic and randomized strategyproof mechanisms changes/improves across the \(L_p\) spectrum for any $p$.

We further show that the robustness of these mechanisms is not limited to
\(L_p\)-norm social costs. For every symmetric monotone norm objective, both
CM and URCM admit a uniform \(2\)-approximation guarantee in the Euclidean
plane. Moreover, the CRD guarantee extends to arbitrary real normed vector
spaces, where it achieves the tight approximation ratio \(2-\frac1n\).

Several questions remain open. 
First, we currently do not have general lower bounds on the approximation ratios of %\greencomment{deterministic non-anonymous strategyproof mechanisms and} 
randomized strategyproof mechanisms under the \(L_p\)-norm social cost. 
Such lower bounds are needed to determine whether the guarantees of the examined strategyproof mechanisms are close to optimal or whether substantially better strategyproof mechanisms exist. 
Second, the upper bounds on the approximation ratios of the 
two randomized strategyproof mechanisms are not always known to be tight. 
It would be interesting either to improve these analyses or to identify tight instances, especially in the intermediate regime \(1<p<2\) for the uniformly rotated coordinate-wise median mechanism and for finite \(p\) in the centroid random dictatorship mechanism. 
Finally, our work focuses on the Euclidean plane. %space. 
An important direction is to study analogous questions in higher-dimensional or non-Euclidean spaces, such as facilities and agents embedded in $q$-normed spaces, and to understand how the geometry of the underlying metric affects both strategyproofness and approximation ratios.

\bibliographystyle{plain}
\bibliography{myreferences}

\end{document}